\documentclass[10pt]{article} 
\usepackage{amsmath,amsfonts,latexsym}
\usepackage[noend]{algpseudocode}
\usepackage{xcolor,algorithm,multirow,lscape,graphicx}	
\usepackage{tikz}
\usepackage{pgfgantt}
\usepackage{hyperref}
\usetikzlibrary{shapes,arrows,fit,calc,positioning}
\usetikzlibrary{decorations.pathreplacing}
\normalsize
\newtheorem{theorem}{Theorem}
\newtheorem{lemma}{Lemma}

\newtheorem{remark}{Remark}

\newtheorem{definition}{Definition}

\newenvironment{proof}{{\sc Proof. }}{\hfill$\Box$\vspace{0.1in}}
\def\mcA{\mathcal{A}}

\title{Maximizing social welfare among EF1 allocations at the presence of two types of agents\footnote{A shorter version appears in {\em Proceedings of ISAAC 2025}.}}
\author{
Jiaxuan~Ma\thanks{Department of Mathematics, Hangzhou Dianzi University. Hangzhou, China.
	Emails: \texttt{\{221070019, chenyong, anzhang\} @hdu.edu.cn}}
\and
Yong~Chen$^\dagger$\thanks{Corresponding authors.}
\and
Guangting~Chen\thanks{Zhejiang University of Water Resources and Electric Power. Hangzhou, China.
	Email: \texttt{gtchen@hdu.edu.cn}}
\and
Mingyang~Gong\thanks{Department of Computing Science, University of Alberta.  Edmonton, Canada.
	Emails: \texttt{\{mgong4, guohui\} @ualberta.ca}}
\and
Guohui~Lin$^{\P,\ddagger}$
\and
An~Zhang$^\dagger$
}
\date{}
\begin{document}
\maketitle
\begin{abstract}
We study the fair allocation of indivisible items to $n$ agents to maximize the utilitarian social welfare,
where the fairness criterion is envy-free up to one item and there are only two different utility functions shared by the agents.
We present a $2$-approximation algorithm when the two utility functions are normalized,
improving the previous best ratio of $16 \sqrt{n}$ shown for general normalized utility functions;
thus this constant ratio approximation algorithm confirms the APX-completeness in this special case previously shown APX-hard.
When there are only three agents, i.e., $n = 3$, the previous best ratio is $3$ shown for general utility functions,
and we present an improved and tight $\frac 53$-approximation algorithm when the two utility functions are normalized,
and a best possible and tight $2$-approximation algorithm when the two utility functions are unnormalized.

\end{abstract}

\section*{Acknowledgments}
The research is supported by the NSERC Canada,
the NNSF of China (Grants 12471301, 12371316),
the PNSF of Zhejiang China (Grant LZ25A010001), and
the China Scholarship Council (Grant No. 202508330173).

\section{Introduction}
We study a special case of the fair allocation of indivisible items to a number of agents to maximize their utilitarian social welfare,
where the agents fall into two types and the fairness criterion is envy-free up to one item.
We present a set of approximation algorithms that each improves the previous best one when reduced to this special case, or is the best possible for this special case.

Throughout this paper, we denote by $\{a_1, a_2, \ldots, a_n\}$ the set of $n$ agents and by $M = \{g_1, g_2, \ldots, g_m\}$ the set of $m$ indivisible items or goods.
An {\em allocation} of items is an $n$-tuple $\mcA = (A_1, A_2, \ldots, A_n)$,
where $A_i$ is the pair-wise disjoint bundle/subset of items allocated/assigned to $a_i$.
The allocation is {\em complete} if $\cup_{i=1}^n A_i = M$, i.e., every item is allocated, or otherwise it is {\em partial}.
Each agent $a_i$ has a non-negative {\em additive} {\em utility} (or {\em valuation}) function $u_i: M \to \mathbb{R}_{\ge 0}$.
Extending to $2^M$, for any subset $S \subseteq M$, $u_i(S) = \sum_{g \in S} u_i(g)$.
If $u_i(M) = 1$ for all $i$, then the utility functions are said {\em normalized};
otherwise, they are {\em unnormalized}.
The {\em utilitarian social welfare} of the allocation $\mcA = (A_1, A_2, \ldots, A_n)$ is defined as $\sum_{i=1}^n u_i(A_i)$.

The allocation $\mcA = (A_1, A_2, \ldots, A_n)$ is {\em envy-free up to one item} (EF1), if for any $i \ne j$, 
$u_i(A_i) \ge u_i(A_j \setminus g)$ holds for some item $g \in A_j$.
Between two agents $a_i$ and $a_j$, if $u_i(A_i) \ge u_i(A_j)$, then $a_i$ does not envy $a_j$;
otherwise, $a_i$ envies $a_j$.
Similarly, if $u_i(A_i) \ge u_i(A_j \setminus g)$ for some $g \in A_j$, then $a_i$ does not {\em strongly} envy $a_j$;
otherwise, $a_i$ strongly envies $a_j$.

In this paper, we study the optimization problem to find an EF1 allocation with the maximum utilitarian social welfare, denoted as USW-EF1,
and we focus on the special case where there are only two types of agents, or equivalently speaking only two distinct utility functions shared by the agents.
W.l.o.g., we assume that the first $j$ agents $a_1, a_2, \ldots, a_j$ use the utility function $u_1$, referred to as the first type of agents,
and the last $n-j$ agents $a_{j+1}, a_{j+2}, \ldots, a_n$ use the second utility function $u_n$, for some $1 \le j < n$.
Let $\mcA$ be the allocation produced by an approximation algorithm and $\mcA^*$ be an optimal allocation.
We denote by $SOL = \sum_{i=1}^n u_i(A_i)$ and $OPT = \sum_{i=1}^n u_i(A^*_i)$ their objective values, respectively.
The algorithm is an $\alpha$-approximation if the ratio $OPT / SOL \le \alpha$ for all instances.
We remark that this two-types-of-agents special case is important for several reasons,
for example, most hardness and inapproximability results are proven for this special case,
and this special case also models some real application scenarios such as individual versus business investors and faculty members versus administration staff.

\subsection{Related work}
{\em Discrete fair division} aims to allocate a set of indivisible items to a set of agents so that
all the agents feel ``fair'' with respect to their heterogeneous preferences over the items.
Typical fairness criteria include {\em envy-freeness} (EF)~\cite{Var74} and {\em proportionality}~\cite{Ste49}.
By an EF allocation, each agent $a_i$ does not envy any other agent $a_j$.
Since an EF allocation does not always exist, one relaxation of EF is envy-freeness up to one item (EF1) proposed in \cite{LMM04,B11}.
Formally speaking, for each pair of agents $a_i, a_j$, agent $a_i$ does not envy $a_j$ after the removal of {\em some} item from the bundle allocated to $a_j$.
Replacing the quantifier ``for some'' by ``for any'' leads to another more restricted relaxation {\em envy-freeness up to any item} (EFX)~\cite{CKM19}.

A complete EF1 allocation can be computed in polynomial time by the {\em envy-cycle elimination} (ECE) algorithm proposed by Lipton et al.~\cite{LMM04} or
by the {\em round-robin} (RR) algorithm proposed by Caragiannis et al.~\cite{CKM19}.
In brief, in the ECE algorithm, an {\em envy digraph} is constructed on the agents where agent $a_i$ {\em points} to agent $a_j$ (i.e., the arc $(a_i, a_j)$ exists) if $a_i$ envies $a_j$;
whenever there is a cycle in the digraph, the bundles allocated to the agents on the cycle are rotated to eliminate the cycle;
afterwards the algorithm assigns an unallocated item to an agent not envied by anyone else, and the iteration goes on.
In the RR algorithm, the agents are ordered cyclically and the agent at the head of the order picks its most preferred item from
the pool of unassigned items and then lines up at the back of the order, until no item is left.

Besides the existence of EF1 allocations,
the price of such fair allocations is an interesting concern in economics and social science.
Given that the utilitarian social welfare measures the {\em efficiency} of an allocation,
the {\em price of EF1} is defined as the ratio between the maximum utilitarian social welfare across all allocations and the maximum utilitarian social welfare among only EF1 allocations.
The price of EF1 for normalized utility functions is $\Theta(\sqrt{n})$:
Barman et al.~\cite{BBS20} first showed that the price is at most $O(\sqrt{n})$,
and later Bei et al.~\cite{BLM21} showed that the price is at least $\Omega(\sqrt{n})$.
The price of EF1 for unnormalized utility functions is $n$:
Barman et al.~\cite{BBS20} showed that the price is $\Theta(n)$, and later Bu et al.~\cite{BLL25} constructed an instance to show that it is at least $n$.

For the optimization problem USW-EF1,
Barman et al.~\cite{BBS20} proved that it is already NP-hard for two agents even when the valuation function of an agent is a scaled version of that of the other,
and it is NP-hard to approximate within a factor of $m^{1-\epsilon}$ for any $\epsilon > 0$ for $n$ agents and $m$ items
when both $n$ and $m$ are part of the input.
When the utility functions are normalized, the problem remains NP-hard for $n$ agents when $n \ge 2$ is a fixed integer~\cite{AHM23,BLL25}.
Moreover, Bu et al.~\cite{BLL25} showed several hardness results when $n \ge 3$ that,
firstly the problem is NP-hard to approximate within the factor $\frac {4n}{3n+1}$ for normalized functions
when one agent uses a utility function and the other agents use a common second utility function;
secondly it is NP-hard to approximate within the factor $\frac {1 + \sqrt{4n-3}}2$ for unnormalized functions
when one agent uses a utility function, two other agents use a common second utility function, and all the other agents (if any) use a common third utility function;
and thirdly it is NP-hard to approximate within the factor $m^{\frac 12 - \epsilon}$ or within the factor $n^{\frac 13 - \epsilon}$, for any $\epsilon > 0$,
for normalized utility functions when both $n$ and $m$ are part of the input.

On the positive side, Barman et al.~\cite{BBS20} presented a $16\sqrt{n}$-approximation algorithm when $n$ is a fixed integer and the utility functions are normalized.
Aziz et al.~\cite{AHM23} proposed a pseudo-polynomial time exact algorithm for the same variant.
Bu et al.~\cite{BLL25} gave a {\em fully polynomial-time approximation scheme} (FPTAS) for two agents and
an $n$-approximation algorithm for $n \ge 3$ agents with unnormalized utility functions.

We remark that we reviewed in the above only those directly related work, but not the entire body of the literature on fair division of indivisible goods.
For example, the Nash social welfare (NSW) objective, the geometric mean of the agents' utilities, has been extensively studied,
and it is known that an allocation which maximizes NSW is EF1~\cite{CKM19} though approximate solutions are not necessarily EF1.
One may refer to \cite{ABF22,AAB23} for an excellent survey on recent progress and open questions.

\subsection{Our contribution and organization of the paper}
In this paper, we aim to design approximation algorithms for the special case of USW-EF1 where there are only two distinct utility functions shared by the agents,
i.e., there are two types of agents.
Note that when all agents share the same utility function, a complete EF1 allocation returned by the ECE or RR algorithm is optimal.
Also, given that the above two lower bounds on the approximation ratios~\cite{BLL25} are proven for two or three utility functions,
our study on the special case may shed lights on the general case.
In particular, we demonstrate the use of the {\em item preference order} in all our three approximation algorithms.

We first present in the next section a $2$-approximation algorithm for any number of agents with normalized utility functions.
Then, in Section 3, we present a tight $2$-approximation algorithm for three agents with unnormalized utility functions,
which is the {\em best possible} by the lower bound $\frac {1 + \sqrt{4n-3}}2$~\cite{BLL25}.
Section 4 contains an improved and tight $\frac 53$-approximation algorithm for three agents with normalized utility functions.
We conclude our paper in Section 5.

In all our three algorithms, the items are firstly sorted, and thus we assume w.l.o.g. that they are given, in the non-increasing order of their {\em preferences},
where the preference of an item $g$ is defined as $u_1(g)/u_n(g)$ (Definition~\ref{def01}) that measures the extent of preference of the first type of agents over the second type of agents.
Intuitively, in an efficient allocation, items at the front of this order should be allocated to the first type of agents,
while items at the back should be allocated to the second type of agents.
This is exactly the main design idea.

\section{A $2$-approximation for normalized functions}
Recall that agents $a_1, \ldots, a_j$ share the first utility function $u_1(\cdot)$ and agents $a_{j+1}, \ldots, a_n$ share the second utility function $u_n(\cdot)$.
For each item $g \in M$, we assume w.l.o.g. that its two values $u_1(g), u_n(g)$ are not both $0$.

\begin{definition}
\label{def01}
The {\em preference} of item $g$ is defined as $\rho(g) = u_1(g)/u_n(g)$, where if $u_n(g)=0$ then $\rho(g)=\infty$.

Extending to a non-empty set of items $S \subseteq M$, the preference of $S$ is defined as $\rho(S) = u_1(S) / u_n(S)$, where if $u_n(S) = 0$ then $\rho(S) = \infty$.

Specially, the preference of an empty set is $\pm\infty$, indicating that it is less than but also greater than any real value.
\end{definition}

We assume w.l.o.g. that $\rho(g_1) \ge \rho(g_2) \ge \ldots \ge \rho(g_m)$, i.e., the items are given in the {\em non-decreasing preference} order.
Our algorithm, denoted as {\sc Approx1}, uses two variables $k_1$ and $k_2$ to store the smallest and the largest indices of the unassigned items,
which are initialized to $1$ and $m$, respectively.
In each iteration, the algorithm finds the agent $a_s$ ($a_t$, resp.) having the minimum utility among the first (second, resp.) type of agents.
({\em Comment: These two agents $a_s$ and $a_t$ will be proven not to envy each other.})
If $a_s$ is not envied by $a_t$, then $a_s$ takes the item $g_{k_1}$ and $k_1$ is incremented by $1$;
otherwise, $a_t$ is not envied by $a_s$, $a_t$ takes $g_{k_2}$ and $k_2$ is decremented by $1$.
The algorithm terminates when all items are assigned (i.e., $k_1 > k_2$), of which a high-level description is depicted in Algorithm~\ref{Approx1}.

\begin{algorithm}
\caption{{\sc Approx1} for normalized utility functions}
\label{Approx1}
{\bf Input:} $n \ge 3$ agents of two types and a set of $m$ indivisible items $\rho(g_1) \ge \rho(g_2) \ge \ldots \ge \rho(g_m)$.

{\bf Output:} A complete EF1 allocation.
\begin{algorithmic}[1]
\State Initialize $k_1=1$, $k_2=m$, and $A_i = \emptyset$ for every agent $a_i$;

\While {($k_1 \le k_2$)}
\State find $s = \arg \min_{i=1}^j u_1(A_i)$ and $t = \arg \min_{i=j+1}^n u_n(A_i)$;

\If {($a_s$ is not envied by $a_t$)}
\State $A_s = A_s \cup \{g_{k_1}\}$ and $k_1 = k_1 + 1$;
\Else
\State $A_t = A_t \cup \{g_{k_2}\}$ and $k_2 = k_2 - 1$;
\EndIf
\EndWhile

\State return the final allocation.
\end{algorithmic}
\end{algorithm}

We next prove that inside each iteration, the found two agents $a_s$ and $a_t$ do not envy each other, and thus {\sc Approx1} produces a complete EF1 allocation.
To this purpose, we introduce the following definition.

\begin{definition}
\label{def02}
For two item sets $A, B \subseteq M$,
if $\min_{g \in A\setminus B} \rho(g) \ge \max_{g \in B \setminus A} \rho(g)$, then we say $A$ {\em precedes} $B$ and denote it as $A \prec B$.
That is, excluding the common items, every item of $A$ (if any) comes before any item of $B$ (if any) in the item preference order.

An allocation $\mcA$ is {\em good} if $A_s \prec A_t$ for every pair $(s, t)$ with $s \le j < t$.
({\em Comment: Since $A_s \cap A_t = \emptyset$, this implies $\rho(A_s) \ge \rho(A_t)$.})

Recall that $\rho(\emptyset) = \pm\infty$, if $A_s = \emptyset$ then both $A_s \prec A_t$ and $A_t \prec A_s$ hold.
\end{definition}

\begin{lemma}
\label{lemma03}
Given a good allocation $\mcA$, two agents of different types do not mutually envy each other.
\end{lemma}
\begin{proof}
Assume to the contrary that for a pair $(s, t)$ with $s \le j < t$, $a_s$ and $a_t$ envy each other, that is, $u_1(A_s) < u_1(A_t)$ and $u_n(A_t) < u_n(A_s)$.
It follows that none of $A_s$ and $A_t$ is $\emptyset$ and $\rho(A_s) = u_1(A_s)/u_n(A_s) < u_1(A_t)/u_n(A_t) = \rho(A_t)$, which contradicts the definition of a good allocation.
\end{proof}

Note that there are exactly $m$ iterations inside the algorithm {\sc Approx1},
and we let $\mcA^0= (\emptyset, \ldots, \emptyset)$ and let $\mcA^i = (A_1^i, \ldots, A_n^i)$ denote the allocation at the {\em end} of the $i$-th iteration.
The final produced allocation is $\mcA = \mcA^m$.

\begin{lemma}
\label{lemma04}
The allocation $\mcA^i$ is good and EF1, for each $i = 0, 1, \ldots, m$.
\end{lemma}
\begin{proof}
We prove by induction.

The initial empty allocation $\mcA^0$ is clearly good and EF1.
Assume $\mcA^i$ is good and EF1 for some $i < m$, and let $a_s$ and $a_t$ be the agents found in the $(i+1)$-st iteration.

Consider the case where $a_s$ is not envied by $a_t$ (the other case is symmetric), in which the algorithm updates $A_s$ to $A_s \cup \{ g_{k_1} \}$.
By the description of {\sc Approx1}, $A_s \cup \{ g_{k_1} \}$ is a non-empty subset of $\{ g_1, \ldots, g_{k_1} \}$
and $A_i \subseteq \{ g_{k_2+1}, \ldots, g_m \}$ for every $i \ge j+1$.
By Definitions~\ref{def01} and~\ref{def02} and $k_1 \le k_2$, $A_s \cup \{ g_{k_1} \} \prec A_i$ for every $i \ge j+1$, and thus $\mcA^{i+1}$ is good.

Since at the beginning of the iteration $a_s$ is not envied by $a_t$, $u_n(A_t) \ge u_n(A_s)$.
Also, by the definition of $s$ and $t$, $u_1(A_i) \ge u_1(A_s)$ for every $i \le j$ and $u_n(A_i) \ge u_n(A_t) \ge u_n(A_s)$ for every $i \ge j+1$.
That is, no agent envies $a_s$ in the allocation $\mcA^i$, and thus does not strongly envy $a_s$ in the allocation $\mcA^{i+1}$ (by removing the item $g_{k_1}$).
Note that in this iteration only $a_s$ gets an item.
Therefore, $\mcA^{i+1}$ is EF1 as $\mcA^i$ is EF1.
\end{proof}

\begin{theorem}
\label{thm05}
{\sc Approx1} is a $2$-approximation algorithm.
\end{theorem}
\begin{proof}
For the optimal allocation $\mcA^*$, we have 
\[
OPT = \sum_{i=1}^n u_i(A^*_i) =  u_1(\cup_{i=1}^j A^*_i) +  u_n(\cup_{i=j+1}^n A^*_i) \le u_1(M) + u_n(M) = 2.
\]
Note from Lemma~\ref{lemma04} that the final allocation $\mcA$ produced by {\sc Approx1} is complete, good and EF1.
By Definition~\ref{def02}, we have $\rho(\cup_{i=1}^j A_i) \ge \rho(\cup_{i=j+1}^n A_i)$.
Therefore, either $\rho(\cup_{i=1}^j A_i) \ge 1$ or $\rho(\cup_{i=j+1}^n A_i) < 1$;
or equivalently, either $u_1(\cup_{i=1}^j A_i) \ge u_n(\cup_{i=1}^j A_i)$ or $u_1(\cup_{i=j+1}^n A_i) < u_n(\cup_{i=j+1}^n A_i)$.
In the first case, we have
\[
SOL = \sum_{i=1}^n u_i(A_i) =  u_1(\cup_{i=1}^j A_i) +  u_n(\cup_{i=j+1}^n A_i)
	=u_1(\cup_{i=1}^j A_i) - u_n(\cup_{i=1}^j A_i) + u_n(M) \ge 1.
\]
In the second case, we have
\[
SOL =  u_1(\cup_{i=1}^j A_i) +  u_n(\cup_{i=j+1}^n A_i)
	=u_1(M) - u_1(\cup_{i=j+1}^n A_i) + u_n(\cup_{i=j+1}^n A_i) > 1.
\]
Therefore, $SOL \ge \frac 12 OPT$ always holds.
\end{proof}

\section{A $2$-approximation for three agents with unnormalized functions}
We continue to assume w.l.o.g. that the items are given in their preference order, that is, $\rho(g_1) \ge \rho(g_2) \ge \ldots \ge \rho(g_m)$,
and use the definitions of {\em precedence} and {\em good allocation} in Definition~\ref{def02}.
Since there are only three agents categorized into two types, we assume w.l.o.g. that $j = 1$,
i.e., agent $a_1$ is of the first type and agents $a_2$ and $a_3$ are of the second type.

For ease of presentation we partition the items into two sets based on their preference:
\begin{equation}
\label{eq01}
X = \{ g \in M: u_1(g) \ge u_3(g) \}, \ \ 
Y = \{ g \in M: u_1(g) < u_3(g) \}.
\end{equation}
We first examine the structure of a fixed optimal EF1 allocation, denoted as $(A^*_1, A^*_2, A^*_3)$.
Denote the item $g^* = \arg \max_{g \in A^*_1} u_3(g)$, i.e., the most valuable to $a_2$ and $a_3$ among those items allocated to $a_1$.
The following lemma establishes an important upper bound on $u_3(A^*_1 \setminus g^*)$.

\begin{lemma}
\label{lemma06}
$u_3(A^*_1 \setminus g^*) \le \frac 13 u_3(M \setminus g^*)$.
\end{lemma}
\begin{proof}
Suppose to the contrary that $u_3(A^*_1 \setminus g^*) > \frac 13 u_3(M \setminus g^*)$.
Then we have
\[
u_3(A^*_2 \cup A^*_3) = u_3(M) - u_3(A_1^*) = u_3(M \setminus g^*) - u_3(A_1^* \setminus g^*) < 2 u_3(A_1^* \setminus g^*).
\]
It follows that at least one of $u_3(A^*_2)$ and $u_3(A^*_3)$ is less than $u_3(A^*_1 \setminus g^*)$.
W.l.o.g., we assume $u_3(A^*_2) < u_3(A^*_1 \setminus g^*)$, and thus by the definition of $g^*$ agent $a_2$ strongly envies $a_1$, a contradiction to EF1.
\end{proof}

We next present an upper bound on the optimal total utility $OPT$.
To this purpose, we define what a {\em critical set} is below.
We remark that although the definition relies on the fixed optimal EF1 allocation $\mcA^*$, we can actually compute a critical set for any item,
by the procedure {\sc Algo2} presented in Subsection~3.1.

\begin{definition}
\label{def07}
An item set $K$ is {\em critical for an item $g$} if the following three conditions are satisfied:
\begin{enumerate}
\parskip=0pt
\item
	$g \in K$ and $K \setminus g \subseteq X$;
\item
	$K \setminus g$ $\prec$ $A^*_1 \setminus g$;
\item
	$u_3(K \setminus g) \ge \min \{ u_3(X \setminus g), \frac 13 u_3(M \setminus g) \}$.
\end{enumerate}
\end{definition}

\begin{lemma}
\label{lemma08}
For any critical set $K$ for $g^*$, we have $OPT \le u_1(K)+u_3(M \setminus K)$.
\end{lemma}
\begin{proof}
We partition the items of $A^*_1 \setminus g^*$ into two sets according to Eq.~(\ref{eq01}):
\[
B_1 = (A_1^* \setminus g^*) \cap X, \ \ B_2 = (A^*_1 \setminus g^*) \cap Y.
\]
One sees from $B_2 \subseteq Y$ that $u_1(B_2) \le u_3(B_2)$.
Since $K \setminus g^* \prec A^*_1 \setminus g^*$ and $B_1 \subseteq A^*_1 \setminus g^*$, $K \setminus g^* \prec B_1$ too.
Let $C = (K \setminus g^*) \cap B_1$;
then $\rho(K \setminus \{g^* \cup C\}) \ge \rho(B_1 \setminus C)$.
Since $K \setminus g^* \subseteq X$, $\rho(K \setminus g^*) \ge 1$.
In summary, we have 
%
\[
\rho(K \setminus \{g^* \cup C\}) \ge \max \{ \rho(B_1 \setminus C), 1 \}.
\]
%
By the third condition in Definition~\ref{def07} and Lemma~\ref{lemma06}, we have
\begin{eqnarray*}
u_3(K \setminus \{g^* \cup C\}) &\ge &\min \{ u_3(X \setminus g^*), \frac 13 u_3(M \setminus g^*) \} - u_3(C) \nonumber\\
	&\ge &\min \{ u_3(X \setminus g^*), u_3(A^*_1 \setminus g^*) \} - u_3(C) \nonumber\\
	&\ge &u_3(B_1 \setminus C).
\end{eqnarray*}
%
The above inequalities together give rise to
\begin{eqnarray}
\label{eq2}
u_1(A_1^* \setminus g^*) - u_3(A_1^* \setminus g^*) & \le & u_1(B_1) - u_3(B_1) \nonumber\\
& = & u_1(C)-u_3(C) + u_1(B_1 \setminus C) - u_3(B_1 \setminus C) \nonumber\\
& = & u_1(C)-u_3(C) + u_3(B_1 \setminus C) ( \rho(B_1 \setminus C) - 1) \nonumber\\
& \le & u_1(C)-u_3(C) + u_3(K \setminus \{g^* \cup C\}) ( \rho(K \setminus \{g^* \cup C\}) - 1) \nonumber \\
& = & u_1(K \setminus g^*) - u_3(K \setminus g^*).
\end{eqnarray}
It follows from the above Eq.~(\ref{eq2}) that
\begin{eqnarray*}
OPT & = & u_1(g^*) + u_1(A_1^* \setminus g^*) + u_3 (A_2^* \cup A_3^*) \\
	& \le & u_1(g^*) + u_3(A_1^* \setminus g^*) + u_1(K \setminus g^*) - u_3(K \setminus g^*) + u_3 (A_2^* \cup A_3^*) \\
	& = & u_1(K)+u_3(M \setminus K).
\end{eqnarray*}
This proves the lemma.
\end{proof}

The next lemma states an important property for any allocation in which agent $a_1$ is envied,
by the fact that $a_2$ and $a_3$ share the same utility function $u_3$.

\begin{lemma}
\label{lemma09}
Suppose $\mcA$ is an allocation in which agent $a_1$ is envied. 
If $a_2$ ($a_3$, resp.) is not envied by $a_3$ ($a_2$, resp.), then $a_2$ envies $a_1$, i.e., $u_3(A_2) < u_3(A_1)$ ($a_3$ envies $a_1$, i.e., $u_3(A_3) < u_3(A_1)$, resp.).
\end{lemma}
\begin{proof}
We prove the lemma for $a_2$ (for $a_3$ it is symmetric), that is, agent $a_2$ is not envied by $a_3$, but $a_1$ is envied by at least one of $a_2$ and $a_3$.

If $a_2$ does not envy $a_1$, then $a_3$ envies $a_1$, i.e., $u_3(A_2) \ge u_3(A_1) > u_3(A_3)$, a contradiction to $a_2$ not envied by $a_3$.
\end{proof}

\subsection{Producing a critical set for every item}
Recall that the items are given in the preference order $\rho(g_1) \ge \rho(g_2) \ge \ldots \ge \rho(g_m)$;
and we see from Eq.~(\ref{eq01}) that $X \prec Y$.
For any ordered item set, if the order inherits the given preference order, then it is {\em regular}.

Note from Lemma~\ref{lemma08} that a critical set for the item $g^*$ plays an important role, but we have no knowledge of $A^*_1$ or $g^*$.
Below we construct a critical set for every item $g_i \in M$, in the procedure {\sc Algo2} depicted in Algorithm~\ref{Algo2}.
The key idea is, excluding $g_i$, the critical set matches a prefix of the item preference order.

\begin{algorithm}
\caption{{\sc Algo2} for computing a critical set}
\label{Algo2}
{\bf Input:} The item preference order $S = \langle g_1, \ldots, g_m\rangle$ and an item $g_i \in M$.

{\bf Output:} A critical set $K$ for $g_i$.
\begin{algorithmic}[1]
\State Find the smallest $t$ such that $\sum_{j=1}^t u_3(g_j) \ge \min \{ u_3(X \setminus g_i), \frac 13 u_3(M \setminus g_i) \}$;

\Comment{If $X \setminus g_i = \emptyset$, then $t=0$.}

\If{($t < i$)}
\State output $K = \{g_1, \ldots, g_t, g_i\}$;
\Else
\State find the smallest $t'$ such that $\sum_{j=1}^{t'} u_3(g_j) \ge \min \{ u_3(X), \frac 13 u_3(M \setminus g_i) + u_3(g_i) \}$;
\State output $K = \{g_1, \ldots, g_{t'}\}$.
\EndIf
\end{algorithmic}
\end{algorithm}

\begin{lemma}
\label{lemma10}
{\sc Algo2} produces a critical set $K$ for the item $g_i$.
\end{lemma}
\begin{proof}
Consider the index $t$ in Line 1 of the procedure~{\sc Algo2}.
Clearly, $t \le |X|$.
If $t < i$, then $g_i \in K$ and $K \setminus g_i = \{g_1, \ldots, g_t\}$.
Therefore, $K \setminus g_i \subseteq X$ and $K \setminus g_i \prec A_1^* \setminus g_i$.
Lastly, $u_3(K\setminus g_i) = \sum_{j=1}^t u_3(g_j) \ge \min \{ u_3(X \setminus g_i), \frac 13 u_3(M \setminus g_i) \}$.
All the three conditions in Definition~\ref{def07} are satisfied and thus $K$ is a critical set for $g_i$.

If $t \ge i$, then $t'$ in Line 5 exists and $|X| \ge t' \ge t \ge i$.
Therefore, $g_i \in K$, $K \subseteq X$ and $K \setminus g_i \prec A_1^* \setminus g_i$.
Lastly, $u_3(K \setminus g_i) = \sum_{j=1}^{t'} u_3(g_j) - u_3(g_i) \ge \min \{ u_3(X \setminus g_i), \frac 13 u_3(M \setminus g_i) \}$.
All the three conditions in Definition~\ref{def07} are satisfied and thus $K$ is a critical set for $g_i$.
\end{proof}

Given that the procedure {\sc Algo2} is able to produce a critical set for any item,
by enumerating the items we will have a critical set $K$ for the item $g^*$, and thus by Lemma~\ref{lemma08} to bound the $OPT$.
From $K$, we will also design algorithms to return a complete EF1 allocation of total utility at least half of this bound and thus at least $\frac 12 \cdot OPT$.

In the next three subsections, we deal with the critical set $K$ produced by {\sc Algo2} for an item $g_i$, separately for three possible cases.
Let $k = |K|$.
When $\max_{g \in K} u_1(g) < \frac 12 u_1(K)$, $K$ contains at least three items, i.e., $k \ge 3$;
if the index of the item $g_i$ is $i \ge k$, then we modify the item order $S$ by moving the item $g_i$ to the $(k-1)$-st position,
and denote the new order as $S' = \langle g'_1, \ldots, g'_m \rangle$.
Note that the net effect is, if $i \ge k$, then the items $g_{k-1}, \ldots, g_{i-1}$ are moved one position forward to become $g'_k, \ldots, g'_i$;
otherwise $S'$ is identical to $S$.
In either case, $S' \setminus g'_{k-1}$ is regular and $K$ contains the first $k$ items in the order $S'$, summarized in Remark~\ref{remark11}.

The three distinguished cases are:
\begin{enumerate}
\parskip=0pt
\item
	$\max_{g \in K} u_1(g) \ge \frac 12 u_1(K)$;
\item
	$\max_{g \in K} u_1(g) < \frac 12 u_1(K)$ and $u_3(g'_k) > u_3(M \setminus K)$;
\item
	$\max_{g \in K} u_1(g) < \frac 12 u_1(K)$ and $u_3(g'_k) \le u_3(M \setminus K)$.
\end{enumerate}

\begin{remark}
\label{remark11}
In Cases 2 and 3, i.e., $\max_{g \in K} u_1(g) < \frac 12 u_1(K)$ and thus $k = |K| \ge 3$,
for the new item order $S'$, $g'_k \ne g_i$, $S' \setminus g'_{k-1}$ is regular, $K = \{ g'_1, g'_2, \ldots, g'_k \}$, 
and 
$\sum_{j=1}^{k-1} u_3(g'_j) < \frac 13 u_3(M \setminus g_i) + u_3(g_i)$.
\end{remark}
\begin{proof}
We only need to prove the last inequality $\sum_{j=1}^{k-1} u_3(g'_j) < \frac 13 u_3(M \setminus g_i) + u_3(g_i)$.

If $K$ is outputted in Line 3 of the algorithm {\sc Algo2},
then $k = t+1$ and $\sum_{j=1}^{k-1} u_3(g'_j) = \sum_{j=1}^{t-1} u_3(g_j) + u_3(g_i) < \frac 13 u_3(M \setminus g_i) + u_3(g_i)$, i.e., the inequality holds.
If $K$ is outputted in Line 6 of the algorithm {\sc Algo2} and $i < t' = k$,
then $\sum_{j=1}^{k-1} u_3(g'_j) = \sum_{j=1}^{t'-1} u_3(g_j) < \frac 13 u_3(M \setminus g_i) + u_3(g_i)$, i.e., the inequality holds.
If $K$ is outputted in Line 6 of the algorithm {\sc Algo2} and $i = t = t' = k$,
then $\sum_{j=1}^{k-1} u_3(g'_j) \le \sum_{j=1}^{t-1} u_3(g_j) + u_3(g_i) < \frac 13 u_3(M \setminus g_i) + u_3(g_i)$, i.e., the inequality holds.
This proves the last inequality.
\end{proof}

\subsection{Case 1: $\max_{g \in K} u_1(g) \ge \frac 12 u_1(K)$}
Recall that $K$ is the critical set for the item $g_i$ produced by the procedure {\sc Algo2}.
This case where $\max_{g \in K} u_1(g) \ge \frac 12 u_1(K)$ is considered easy, and we present an algorithm {\sc Approx3} to produce a complete EF1 allocation.

The algorithm starts with the allocation $\mcA = (\{g\}, \emptyset, \emptyset)$, where the item $g$ is one such that $u_1(g) \ge \frac 12 u_1(K)$.
$\mcA$ is trivially EF1 and $u_1(A_1) \ge \frac 12 u_1(K)$.
In fact, $\mcA$ is {\em well-defined} with permutation $(1, 2, 3)$, formally defined below.

\begin{definition}
\label{def12}
An allocation $\mcA$ is {\em well-defined} with permutation $(i_1, i_2, i_3)$ if 
\begin{enumerate}
\parskip=0pt
\item
	$\mcA$ is EF1;
\item
	$u_1(A_1) \ge \frac 12 u_1(K)$ and $u_3(A_1) \le u_3(K)$;
\item
	$(i_1, i_2, i_3) \in \{(1, 2, 3), (3, 1, 2)\}$ such that agent $a_{i_1}$ does not envy $a_{i_2}$ or $a_{i_3}$, and $a_{i_2}$ does not envy $a_{i_3}$.
\end{enumerate}
\end{definition}

Given the well-defined allocation $\mcA = (\{g\}, \emptyset, \emptyset)$ with permutation $(1, 2, 3)$, let $U = M \setminus (A_1 \cup A_2 \cup A_3)$ be the unassigned item set.
The algorithm {\sc Approx3} fixes the agent order $\langle a_3, a_2, a_1 \rangle$ to apply the Round-Robin (RR) algorithm to distribute the items of $U$.
A high-level description of the algorithm is depicted in Algorithm~\ref{Approx3}, which accepts a general well-defined allocation.

\begin{algorithm}
\caption{{\sc Approx3} for a complete EF1 allocation}
\label{Approx3}
{\bf Input:} A well-defined allocation $\mcA = (A_1, A_2, A_3)$ with permutation $(i_1, i_2, i_3)$.

{\bf Output:} A complete EF1 allocation.
\begin{algorithmic}[1]
\State Fix the agent order $\langle a_{i_3}, a_{i_2}, a_{i_1} \rangle$;
\State apply the RR algorithm to allocate the unassigned items, denoted as $(U_1, U_2, U_3)$;
\State output the final allocation $(A_1 \cup U_1, A_2 \cup U_2, A_3 \cup U_3)$.
\end{algorithmic}
\end{algorithm}

\begin{lemma}
\label{lemma13}
Given a well-defined allocation $\mcA = (A_1, A_2, A_3)$ with permutation $(i_1, i_2, i_3)$,
{\sc Approx3} produces a complete EF1 allocation with total utility at least $\frac 12 u_1(K)+\frac 12 u_3(M \setminus K)$.
\end{lemma}
\begin{proof}
Note from the RR algorithm that the returned allocation $(A_1 \cup U_1, A_2 \cup U_2, A_3 \cup U_3)$ is complete.

The initial allocation $\mcA = (A_1, A_2, A_3)$ is EF1 by Definition~\ref{def12}.
Consider any two agents $a_s$ and $a_t$, where $s$ precedes $t$ in the permutation $(i_1, i_2, i_3)$.
That is, $a_s$ does not envy $a_t$, and thus $u_s(A_s) \ge u_s(A_t)$.
Since the RR algorithm uses the agent order $\langle a_{i_3}, a_{i_2}, a_{i_1} \rangle$,
we have $u_s(U_s) \ge u_s(U_t) - u_s(g)$, where $g$ is the first item picked by agent $a_t$;
and thus $u_s(A_s \cup U_s) \ge u_s(A_t \cup U_t) - u_s(g)$ for the same $g \in U_t$.
That is, $a_s$ does not strongly envy $a_t$.
Similarly, since $\mcA$ is EF1, $u_t(A_t) \ge u_t(A_s) - u_t(g)$ for some $g \in A_s$;
we have $u_t(U_t) \ge u_t(U_s)$ by the RR algorithm.
Thus $u_t(A_t \cup U_t) \ge u_t(A_s \cup U_s) - u_t(g)$ for the same $g \in A_s$, that is, $a_t$ does not strongly envy $a_s$.
Therefore, the returned allocation is EF1.

We estimate the total utility of the returned allocation as follows.
Since the agent index $1$ precedes $2$ in the permutation $(i_1, i_2, i_3) \in \{(1, 2, 3), (3, 1, 2)\}$, we have $u_3(U_2) \ge u_3(U_1)$ from the RR algorithm;
we also have $u_1(A_1) \ge \frac 12 u_1(K)$ and $u_3(A_1) \le u_3(K)$ from Definition~\ref{def12} of well-defined allocation,
the total utility of the returned allocation is $u_1(A_1 \cup U_1) + u_3(A_2 \cup A_3 \cup U_2 \cup U_3)$, which is at least
\[
\frac 12 u_1(K) + \frac 12 u_3(M \setminus A_1)
	= \frac 12 u_1(K) + \frac 12 u_3(M) - \frac 12 u_3(A_1)
	\ge \frac 12 u_1(K) + \frac 12 u_3(M \setminus K).
\]
This proves the lemma.
\end{proof}

\begin{theorem}
\label{thm14}
Let $K$ be the critical set of the item $g^*$ produced by the procedure {\sc Algo2}.
If $\max_{g \in K} u_1(g) \ge \frac 12 u_1(K)$, then {\sc Approx3} returns a complete EF1 allocation with its total utility at least $\frac 12 OPT$.
\end{theorem}
\begin{proof}
For this set $K$, the allocation $\mcA = (\{g\}, \emptyset, \emptyset)$ is well-defined with the permutation $(1, 2, 3)$, where $g = \arg\max_{g \in K} u_1(g)$.
The theorem follows from Lemmas~\ref{lemma13} and~\ref{lemma08}.
\end{proof}

\subsection{Case 2: $\max_{g \in K} u_1(g) < \frac 12 u_1(K)$ and $u_3(g'_k) > u_3(M \setminus K)$}
Recall that $K$ is the critical set of an item $g_i$ produced by the procedure {\sc Algo2}.
In this case, $k = |K| \ge 3$ and we work with the new item order $S'$ such that $S' \setminus g'_{k-1}$ is regular, and $g'_k \ne g_i$.
Adding the last inequality $\sum_{j=1}^{k-1} u_3(g'_j) < \frac 13 u_3(M \setminus g_i) + u_3(g_i)$ from Remark~\ref{remark11} and $u_3(g'_k) > u_3(M \setminus K)$ together gives
\begin{equation}
\label{eq03}
u_3(g'_k) > \frac 13 u_3(M \setminus g_i) \ge u_3(K \setminus \{g'_k, g_i\}), \mbox{ and thus } u_3(g'_k) > \frac 12 u_3(K \setminus g_i).
\end{equation}
One sees from the above three lower bounds on $u_3(g'_k)$ that the item $g'_k$ is very valuable to agents $a_2$ and $a_3$.
We design another algorithm called {\sc Approx7} for this case to construct a complete EF1 allocation, in which $g'_k$ is the only item assigned to $a_3$.

Inside the algorithm {\sc Approx7}, the first procedure {\sc Algo4} produces a partial EF1 and good allocation.
It starts with the allocation $\mcA = (\{g'_1\}, \emptyset, \{g'_k\})$, which is EF1 and good (see Definition~\ref{def02}),
and uses the order $S'$ to allocate the items in the two sets $K \setminus \{g'_1 \cup g'_k\}$ and $M \setminus K$ to the agents $a_1$ and $a_2$, respectively,
till exactly one of them becomes empty.
A detailed description of the procedure is depicted in Algorithm~\ref{Algo4}.

\begin{algorithm}
\caption{{\sc Algo4} for a partial EF1 and good allocation}
\label{Algo4}
{\bf Input:} The critical set $K$ for item $g_i$ satisfying $\max_{g \in K} u_1(g) < \frac 12 u_1(K)$,
	the new item order $S'$ in which $g'_k \ne g_i$, and $u_3(g'_k) > u_3(M \setminus K)$.

{\bf Output:} A partial EF1 allocation $\mcA$.
\begin{algorithmic}[1]
\State Initialize $\mcA=(\{g'_1\}, \emptyset, \{g'_k\})$ and $U = M \setminus \{g'_1, g'_k\}$;
\State initialize $k_1=2$ and $k_2=m$;
\While{($k_1 < k < k_2$)}
\If{($a_1$ is not envied by $a_2$)}
\State $A_1 = A_1 \cup \{g'_{k_1}\}$, $U = U \setminus g'_{k_1}$, and $k_1=k_1+1$;
\Else
\State $A_2 = A_2 \cup \{g'_{k_2}\}$, $U = U \setminus g'_{k_2}$, and $k_2=k_2-1$;
\EndIf
\EndWhile
\State return $\mcA$.
\end{algorithmic}
\end{algorithm}

\begin{lemma}
\label{lemma15}
Given the critical set $K$ for item $g_i$ produced by the procedure {\sc Algo2} satisfying $\max_{g \in K} u_1(g) < \frac 12 u_1(K)$, $k = |K|$,
the new item order $S'$ in which $g'_k \ne g_i$, and $u_3(g'_k) > u_3(M \setminus K)$,
{\sc Algo4} produces a partial EF1 and good allocation.
\end{lemma}
\begin{proof}
Note from the description that when {\sc Algo4} terminates,
either $k_1 = k < k_2$ or $k_1 < k = k_2$, i.e., at least one item remains unassigned, and thus the returned allocation $\mcA$ is incomplete.

We remark that since the procedure starts with the allocation $(\{g'_1\}, \emptyset, \{g'_k\})$,
and allocates the items in the two sets $K \setminus \{g'_1 \cup g'_k\}$ and $M \setminus K$ to the agents $a_1$ and $a_2$, respectively,
the bundle $A_1 \prec A_2$ and $A_1 \prec A_3$ throughout the procedure, and thus the allocation maintains good (see Definition~\ref{def02}).
Also throughout the procedure, since $A_3 = \{g'_k\}$ contains a single item, no one strongly envies $a_3$;
since $u_3(g'_k) > u_3(M \setminus K)$ and Eq.~(\ref{eq03}), agent $a_3$ does not strongly envy any of $a_1$ or $a_2$.

We show next that throughout the procedure, between $a_1$ and $a_2$, no one strongly envies the other, and thus the final allocation is EF1.
This is true for the starting allocation $(\{g'_1\}, \emptyset, \{g'_k\})$.
Assume this is true for the allocation at the beginning of an iteration of the while-loop;
since the allocation is good, by Lemma~\ref{lemma03} $a_1$ and $a_2$ do not envy each other.
If $a_1$ is not envied by $a_2$, then the item $g'_{k_1} \in K \setminus g'_k$ is assigned to $a_1$,
and thus $a_2$ does not strongly envy $a_1$ due to $g'_{k_1}$ at the end of the iteration;
$a_1$ remains not strongly envy $a_2$, also due to $g'_{k_1}$ at the end of the iteration.
If $a_2$ is not envied by $a_1$, then the item $g'_{k_2} \in M \setminus K$ is assigned to $a_2$,
and thus $a_1$ does not strongly envy $a_2$ due to $g'_{k_2}$ at the end of the iteration;
$a_2$ remains not strongly envy $a_1$, also due to $g'_{k_2}$ at the end of the iteration.
\end{proof}

If the procedure {\sc Algo4} terminates at $k_1 < k = k_2$, then the returned allocation is $\mcA = (\{g'_1, \ldots, g'_{k_1-1}\}, M \setminus K, \{g'_k\})$
and the unassigned item set is $U = \{g'_{k_1}, \ldots, g'_{k-1}\}$.
The next procedure {\sc Algo5} continues on to assign the items of $U$ to produce a complete EF1 allocation, by re-setting $k_2 = k-1$.
A detailed description of {\sc Algo5} is depicted in Algorithm~\ref{Algo5}.
Basically, if at least one of $a_1$ and $a_2$ is not envied by any of the other two agents,
then the procedure allocates an unassigned item to such a non-envied agent;
otherwise, both $a_1$ and $a_2$ are envied by some agent and the procedure allocates all the remaining unassigned items to $a_1$.
That is, either $g'_{k_1}$ is assigned to $a_1$, or $g'_{k_2}$ is assigned to $a_2$, or all of the items of $U = \{g'_{k_1}, \ldots, g'_{k_2}\}$ are assigned to $a_1$, respectively.
One sees that the allocation $\mcA$ remains good before the last item is assigned
(in fact, the allocation loses being good only if the last item is $g'_{k-1} = g_i$ and it is assigned to $a_1$) in the procedure~{\sc Algo5},
and consequently $a_1$ and $a_2$ do not envy each other by Lemma~\ref{lemma03} before the last item is assigned.

\begin{algorithm}
\caption{{\sc Algo5} for a complete EF1 allocation}
\label{Algo5}
{\bf Input:} The critical set $K$ for item $g_i$ satisfying $\max_{g \in K} u_1(g) < \frac 12 u_1(K)$,
	the new item order $S'$ in which $g'_{k-1} = g_i$, and $u_3(g'_k) > u_3(M \setminus K)$;
	the EF1 and good allocation $\mcA = (\{ g'_1, \ldots, g'_{k_1-1}\}, M \setminus K, \{g'_k\})$ returned by {\sc Algo4}, where $k_1 < k$.

{\bf Output:} A complete EF1 allocation $\mcA$.
\begin{algorithmic}[1]
\State Set $k_2 = k-1$ and the unassigned item set $U = \{g'_{k_1}, \ldots, g'_{k_2}\}$;

\While{($k_1 \le k_2$)}
\If{($a_1$ is not envied)}
\State $A_1 = A_1 \cup \{g'_{k_1}\}$, $U = U \setminus g'_{k_1}$, and $k_1 = k_1+1$;
\ElsIf{($a_2$ is not envied)}
\State $A_2 = A_2 \cup \{g'_{k_2}\}$, $U = U \setminus g'_{k_2}$, and $k_2 = k_2-1$;
\Else
\State $A_1 = A_1 \cup U$, and $k_1 = k_2 + 1$;
\EndIf
\EndWhile

\State return $\mcA$.
\end{algorithmic}
\end{algorithm}

\begin{lemma}
\label{lemma16}
{\sc Algo5} produces a complete EF1 allocation.
\end{lemma}
\begin{proof}
Since there will not be any unassigned item at the end of the procedure, the returned allocation is complete.
We remind the readers that during the execution of {\sc Algo5}, the allocation $\mcA$ remains good before the last item is assigned.

The starting allocation $\mcA = (\{ g'_1, \ldots, g'_{k_1-1}\}, M \setminus K, \{g'_k\})$, which is returned by {\sc Algo4}, is EF1 and good.
At the beginning of an iteration of the while-loop,
if $a_1$ is not envied by any other agent (the case where $a_2$ is not envied by any other agent is symmetric),
then after $g'_{k_1}$ is allocated to $a_1$ no agent will strongly envy $a_1$.
That is, the updated allocation is EF1 too.

Consider the last case of the while-loop, i.e., at the beginning of the last iteration $\mcA = (A_1, A_2, \{g'_k\})$ is the EF1 and good allocation
in which both $a_1$ and $a_2$ are envied by some agent.

We claim that $a_3$ envies $a_2$.
If not, i.e., $u_3(g'_k) \ge u_3(A_2)$, then $a_2$ is envied by $a_1$, and thus by Lemma~\ref{lemma03}, $a_2$ does not envy $a_1$, i.e., $u_3(A_2) \ge u_3(A_1)$.
It follows that $a_3$ does not envy $a_1$.
That is, $a_1$ is not envied by any agent, a contradiction.

In the final allocation $\mcA' = (A_1 \cup U, A_2, \{g'_k\})$, $a_1$ does not strongly envy any of $a_2$ or $a_3$ for sure, and none of $a_2$  and $a_3$ strongly envies the other.
From Eq.~(\ref{eq03}) $a_3$ does not strongly envy $a_1$.
Since $a_3$ envies $a_2$, $a_2$ does not strongly envy $a_1$ either.

In summary, no agent strongly envies another agent in the final allocation $\mcA'$.
\end{proof}

The allocation returned by the procedure~{\sc Algo5} is complete and EF1.
The next procedure~{\sc Algo6} takes in a complete EF1 allocation $\mcA = (A_1, A_2, A_3)$,
and if in which agent $a_1$ envies $a_2$, i.e., $u_1(A_1) < u_1(A_2)$, then outputs the allocation between $\mcA$ and $(A_2, A_1, A_3)$ with a larger total utility.
In general, for a complete EF1 allocation $\mcA = (A_1, A_2, A_3)$, the allocation $(A_2, A_1, A_3)$ after bundle swapping might not be EF1.
We nevertheless show in the next lemma that the allocation returned by the procedure~{\sc Algo6} is EF1.

\begin{algorithm}
\caption{{\sc Algo6} for a larger total utility}
\label{Algo6}
{\bf Input:} The complete EF1 allocation $\mcA = (A_1, A_2, A_3)$ returned by {\sc Algo5}.

{\bf Output:} A complete EF1 allocation.
\begin{algorithmic}[1]
\If{($u_1(A_1) \ge u_1(A_2)$)}
\State output $\mcA$;
\Else
\State output the allocation between $(A_1, A_2, A_3)$ and $(A_2, A_1, A_3)$ with a larger total utility.
\EndIf
\end{algorithmic}
\end{algorithm}

\begin{lemma}
\label{lemma17}
{\sc Algo6} produces a complete EF1 allocation.
\end{lemma}
\begin{proof}
Note that $\mcA = (A_2, A_1, \{g'_k\})$ can be returned only if $u_1(A_1) < u_1(A_2)$.
We show that in such a case, none of $a_1$ and $a_2$ strongly envies the other in $(A_2, A_1, \{g'_k\})$
(noting that $a_3$ does not strongly envy $a_1$ or $a_2$, the same as in $\mcA$, and none of $a_1$ and $a_2$ strongly envies $a_3$ who is allocated with the single item $g'_k$).
Because $u_1(A_2) > u_1(A_1)$, $a_1$ does not envy $a_2$ in the allocation $(A_2, A_1, \{g'_k\})$.

Let $g$ be the last item allocated to $A_2$, done either by {\sc Algo4} or by {\sc Algo5}.
If this is done by {\sc Algo4}, then ($g$ is $g'_{k+1}$ and) by Lines 6--7 in the description at that moment $a_2$ envies $a_1$,
and thus $u_3(A_2 \setminus g) < u_3(\{g'_1, \ldots, g'_{k_1-1}\} \le u_3(A_1)$.
If $g$ is allocated by {\sc Algo5}, then by Lines 5--6 in the description at that moment $a_1$ is envied by some agent and $a_2$ is not envied by any agent,
and thus by Lemma~\ref{lemma09} $u_3(A_2 \setminus g) < u_3(\{g'_1, \ldots, g'_{k_1-1}\} \le u_3(A_1)$.
Therefore, we always have $u_3(A_1) > u_3(A_2 \setminus g)$, i.e., $a_2$ does not strongly envy $a_1$ in the allocation $(A_2, A_1, \{g'_k\})$.
This proves the lemma.
\end{proof}

The algorithm, denoted as {\sc Approx7}, for producing a complete EF1 allocation for Case 2, is depicted in Algorithm~\ref{Approx7},
which utilizes the three procedures introduced in the above.
When {\sc Algo4} terminates at $k_1 = k < k_2$, the returned partial allocation is $\mcA = (K \setminus g'_k, A_2, \{g'_k\})$ where $A_2 \subset M \setminus K$,
which is completed by calling the Envy Cycle Elimination (ECE) algorithm to assign the rest of the items, i.e., $g'_{k+1}, \ldots, g'_{k_2}$.

\begin{algorithm}
\caption{{\sc Approx7} for a complete EF1 allocation}
\label{Approx7}
{\bf Input:} The critical set $K$ for item $g_i$ satisfying $\max_{g \in K} u_1(g) < \frac 12 u_1(K)$,
	the new item order $S'$ in which $g'_k \ne g_i$, and $u_3(g'_k) > u_3(M \setminus K)$.

{\bf Output:} A complete EF1 allocation.
\begin{algorithmic}[1]
\State Call {\sc Algo4} to produce an EF1 partial allocation $\mcA = (A_1, A_2, \{g'_k\})$, and $k_1, k_2$;
\State set $U = M \setminus (A_1 \cup A_2 \cup \{g'_k\})$;

\If{($k_1 = k < k_2$)}
\State call the ECE algorithm on $\mcA$ to continue to assign the items in $U$;
\State output the final allocation;
\Else
\State re-set $k_2 = k - 1$;
\State call {\sc Algo5} on $\mcA$ to continue to assign the items in $U$;
\State call {\sc Algo6} on $\mcA$ to output the final allocation. 
\EndIf
\end{algorithmic}
\end{algorithm}

\begin{lemma}
\label{lemma18}
{\sc Approx7} produces a complete EF1 allocation.
\end{lemma}
\begin{proof}
If the procedure {\sc Algo4} terminates at $k_1 = k < k_2$,
then since the returned partial allocation is EF1 by Lemma~\ref{lemma15}, the ECE algorithm continues from there to produce a complete EF1 allocation~\cite{LMM04}.

If the procedure {\sc Algo4} terminates at $k_1 < k = k_2$,
then the final allocation is returned by {\sc Algo6} and by Lemma~\ref{lemma17} is complete and EF1.
\end{proof}

\begin{theorem}
\label{thm19}
Given the critical set $K$ for the item $g^*$ produced by the procedure {\sc Algo2} satisfying $\max_{g \in K} u_1(g) < \frac 12 u_1(K)$, $k = |K|$,
the new item order $S'$ in which $g'_k \ne g^*$, and $u_3(g'_k) > u_3(M \setminus K)$,
{\sc Approx7} constructs a complete EF1 allocation with its total utility at least $\frac 12 OPT$.
\end{theorem}
\begin{proof}
By Lemma~\ref{lemma18} the final allocation is complete and EF1.

If the procedure {\sc Algo4} terminates at $k_1 = k < k_2$, then $SOL \ge u_1(K \setminus g'_k) + u_3(g'_k)$.
Since $u_1(g'_k) < \frac 12 u_1(K)$ and $u_3(g'_k) > u_3(M \setminus K)$,
the total utility is greater than $\frac 12 u_1(K) + u_3(M \setminus K) \ge \frac 12 OPT$ where the last inequality holds by Lemma~\ref{lemma08}.

If the procedure {\sc Algo4} terminates at $k_1 < k = k_2$, then let $\mcA = (A_1, A_2, \{g'_k\})$ denote the allocation returned by the procedure {\sc Algo5},
in which $A_1 \subseteq K \setminus g'_k$ and $A_2 \subset M \setminus g'_k$.
We distinguish to whom $g'_{k-1}$ is allocated.

If $g'_{k-1} \in A_1$, then $A_1 = K \setminus g'_k$, and the above argument for the first termination condition applies too, so that $SOL > \frac 12 OPT$,
by noting that the procedure~{\sc Algo6} never reduces the total utility.

If $g'_{k-1} \in A_2$, then $(M \setminus K) \cup \{g'_{k-1}, g'_k\} \subseteq A_2 \cup A_3$.
By $\sum_{j=1}^{k-1} u_3(g'_j) < \frac 13 u_3(M \setminus g^*) + u_3(g^*)$ from Remark~\ref{remark11}, $u_3(A_2) + u_3(A_3) > \frac 23 u_3(M \setminus g^*)$.
When $u_1(A_1) < u_1(A_2)$, the procedure {\sc Algo6} outputs $\mcA$ or $(A_2, A_1, A_3)$ with a larger total utility, which is
\[
SOL \ge \frac 12 u_1(A_1 \cup A_2) + \frac 12 u_3(A_1 \cup A_2) + u_3(A_3) \ge \frac 12 u_1(M \setminus g'_k) + \frac 12 u_3(M).
\]
When $u_1(A_1) \ge u_1(A_2)$, $\mcA$ is the final allocation with its total utility
\[
SOL \ge \frac 12 u_1(A_1 \cup A_2) + u_3(A_2) + u_3(A_3) > \frac 12 u_1(M \setminus g'_k) + \frac 23 u_3(M \setminus g^*).
\]
Noting from $g'_k \ne g^*$, $u_3(g'_k) > \frac 13 u_3(M \setminus g^*)$ in Eq.~(\ref{eq03}), and $u_3(A^*_1 \setminus g^*) \le \frac 13 u_3(M \setminus g^*)$ in Lemma~\ref{lemma06},
we have $g'_k \notin A^*_1$.
Therefore, $u_1(A^*_1) \le u_1(M \setminus g'_k)$ and then since $g^* \in A^*_1$,
\[
SOL \ge \frac 12 u_1(M \setminus g'_k) + \min\{\frac 12 u_3(M), \frac 23 u_3(M \setminus g^*)\} \ge \frac 12 u_1(A^*_1) + \frac 12 u_3(A^*_2 \cup A^*_3) = \frac 12 OPT.
\]
This proves the theorem.
\end{proof}

\subsection{Case 3: $\max_{g \in K} u_1(g) < \frac 12 u_1(K)$ and $u_3(g'_k) \le u_3(M \setminus K)$}
Recall that $K$ is the critical set of item $g_i$ produced by the procedure {\sc Algo2}.
The same as in Case 2, here we also have $k = |K| \ge 3$ and we work with the new item order $S'$ such that $K = \{g'_1, \ldots, g'_k\}$,
$S' \setminus g'_{k-1}$ is regular, and $g'_k \ne g_i$.
Note the difference from Case 2 that, here we have $u_3(g'_k) \le u_3(M \setminus K)$.
We design an algorithm denoted {\sc Approx9} for this case to construct a complete EF1 allocation.

Inside the algorithm {\sc Approx9}, the procedure {\sc Algo8} produces a partial EF1 allocation.
Different from {\sc Algo4} for Case 2, {\sc Algo8} starts with the allocation $\mcA = (\{g'_1\}, \{g'_k\}, \emptyset)$, which is EF1 and good (see Definition~\ref{def02}),
and uses the order $S'$ to allocate the items in the set $K \setminus \{g'_1 \cup g'_k\}$ to the agents $a_1$ and $a_2$, and the items in the set $M \setminus K$ to agent $a_3$,
respectively, till exactly one of the two sets becomes empty.
A detailed description of the procedure is depicted in Algorithm~\ref{Algo8}.

\begin{algorithm}
\caption{{\sc Algo8} for a partial EF1 allocation}
\label{Algo8}
{\bf Input:} The critical set $K$ for item $g_i$ satisfying $\max_{g \in K} u_1(g) < \frac 12 u_1(K)$,
	the new item order $S'$ in which $g'_k \ne g_i$, and $u_3(g'_k) \le u_3(M \setminus K)$.

{\bf Output:} A partial EF1 allocation $\mcA$.
\begin{algorithmic}[1]
\State Initialize $\mcA = (\{g'_1\}, \{g'_k\}, \emptyset)$ and $U = M \setminus \{g'_1, g'_k\}$;
\State initialize $k_1 = 2$, $k_2 = k-1$, and $k_3 = m$;
\While{($k_1 \le k_2$ and $k_3 \ge k+1$)}
	\If{($a_1$ is not envied)}
	\State $A_1 = A_1 \cup \{g'_{k_1}\}$, $U = U \setminus g'_{k_1}$, and $k_1 = k_1 + 1$;
	\ElsIf{($a_2$ is not envied)}
	\State $A_2 = A_2 \cup \{g'_{k_2}\}$, $U = U \setminus g'_{k_2}$, and $k_2 = k_2 - 1$;
	\Else \Comment{in this case $a_3$ is not envied}
	\State $A_3 = A_3 \cup \{g'_{k_3}\}$, $U = U \setminus g'_{k_3}$, and $k_3 = k_3 - 1$;
	\EndIf
\EndWhile
\State return $\mcA$.
\end{algorithmic}
\end{algorithm}

\begin{lemma}
\label{lemma20}
Given the critical set $K$ for item $g_i$ produced by the procedure {\sc Algo2} satisfying $\max_{g \in K} u_1(g) < \frac 12 u_1(K)$, $k = |K|$,
the new item order $S'$ in which $g'_k \ne g_i$, and $u_3(g'_k) \le u_3(M \setminus K)$,
{\sc Algo8} produces a partial EF1 allocation.
The allocation is good unless item $g'_{k-1}$ is allocated to agent $a_1$, and this happens in the last iteration of the while-loop.
\end{lemma}
\begin{proof}
Note from when {\sc Algo8} terminates, either $k_1 \le k_2$ or $k_3 \ge k+1$, i.e., at least one item remains unassigned, and thus the returned allocation $\mcA$ is incomplete.

We remark that since the procedure starts with the allocation $(\{g'_1\}, \{g'_k\}, \emptyset)$,
and allocates the items in the set $K \setminus \{g'_1 \cup g'_k\}$ to the agents $a_1$ and $a_2$, and the items in the set $M \setminus K$ to agent $a_3$, respectively,
the bundle $A_1 \prec A_2$ and $A_1 \prec A_3$ (i.e., the allocation is good) throughout the procedure unless the item $g'_{k-1}$ is allocated to $a_1$
(in fact, only if $g'_{k-1} = g_i$) and in such a scenario {\sc Algo2} terminates right afterwards.

We show that throughout the procedure, the allocation is EF1.
This is true for the starting allocation $(\{g'_1\}, \{g'_k\}, \emptyset)$.
Assume this is true for the allocation $\mcA = (A_1, A_2, A_3)$ at the beginning of an iteration of the while-loop.
If $a_1$ is not envied, then the item $g'_{k_1}$ is assigned to $a_1$, and thus none of $a_2$ and $a_3$ strongly envies $a_1$ due to $g'_{k_1}$ at the end of the iteration.
If $a_2$ is not envied, then the item $g'_{k_2}$ is assigned to $a_2$, and thus none of $a_1$ and $a_3$ strongly envies $a_2$ due to $g'_{k_2}$ at the end of the iteration.
Otherwise, the item $g'_{k_3}$ is assigned to $a_3$.

We claim that $a_3$ is not envied.
Assume $a_3$ is envied by $a_1$, i.e., $u_1(A_1) < u_1(A_3)$.
Since $\mcA$ is good, by Lemma~\ref{lemma03} $a_1$ is not envied by $a_3$ and thus $a_1$ is envied by $a_2$, i.e., $u_3(A_2) < u_3(A_1) \le u_3(A_3)$.
Therefore, $a_2$ is not envied by $a_3$, and thus $a_2$ is envied by $a_1$, a contradiction to $\mcA$ being good.
Assume $a_3$ is envied by $a_2$, i.e., $u_3(A_2) < u_3(A_3)$.
Then $a_2$ is not envied by $a_3$ and thus $a_2$ is envied by $a_1$, i.e., $u_1(A_1) < u_1(A_2)$.
Since $a_1$ is envied, from either $u_3(A_2) < u_3(A_1)$ or $u_3(A_3) < u_3(A_1)$, we have $u_3(A_2) < u_3(A_1)$, i.e., $a_1$ is envied by $a_2$, a contradiction to $\mcA$ being good.

Now since $a_3$ is not envied, none of $a_1$ or $a_2$ strongly envies $a_3$ due to $g'_{k_3}$ at the end of the iteration.
Therefore, the final allocation returned by {\sc Algo8} is EF1.
\end{proof}

There are two possible termination conditions, $k_3 = k$ or $k_1 > k_2$.
If $k_3 = k$, then the allocation returned by {\sc Algo8} is $\mcA = (A_1, A_2, M\setminus K)$, which is good and in which $A_1 \cup A_2 \subset K$,
and the unassigned item set is $U = \{g_{k_1}, \ldots, g_{k_2}\}$ with $k_1 \le k_2 \le k-1$.
One sees that if the item $g'_{k-1}$ has been allocated, then it must be allocated to agent $a_2$.
Therefore, $u_3((A_1 \cup U) \setminus g_i) \le \sum_{j=1}^{k-1} u_3(g'_j) - u_3(g_i) < \frac 13 u_3(M \setminus g_i)$ by Remark~\ref{remark11}.
Also by $\sum_{j=1}^{k-1} u_3(g'_j) < \frac 13 u_3(M \setminus g_i) + u_3(g_i)$ from Remark~\ref{remark11} and $u_3(g'_k) \le u_3(M \setminus K)$, we have
\begin{equation}
\label{eq04}
u_3(M \setminus K) > \frac 13 u_3(M \setminus g_i) > u_3((A_1 \cup U) \setminus g_i).
\end{equation}
Similar to the algorithm~{\sc Approx7} for Case 2,
our algorithm {\sc Approx9} first calls the procedure {\sc Algo5} to continue from the partial EF1 and good allocation $\mcA$ to assign the items of $U$ to the agents $a_1$ and $a_2$,
and then calls {\sc Algo6} to output the final allocation.

If {\sc Algo8} terminates at $k_1 > k_2$, then the allocation returned by {\sc Algo8} is $\mcA = (A_1, A_2, A_3)$, where $A_1 \cup A_2 = K$ and $A_3 \subset M \setminus K$,
and the unassigned item set is $U = \{g'_{k+1}, \ldots, g'_{k_3}\} \subset M \setminus K$.
From $\mcA$, we show that the three bundles can be shuffled to become a well-defined allocation with a permutation either $(1, 2, 3)$ or $(3, 1, 2)$ (see Definition~\ref{def12}).
Lastly, the algorithm~{\sc Approx3} is called on the achieved well-defined allocation with the permutation to produce the final allocation.
A high-level description of the complete algorithm {\sc Approx9} is depicted in Algorithm~\ref{Approx9}.

\begin{algorithm}
\caption{{\sc Approx9} for a complete EF1 allocation}
\label{Approx9}
{\bf Input:} The critical set $K$ for item $g_i$ satisfying $\max_{g \in K} u_1(g) < \frac 12 u_1(K)$,
	the new item order $S'$ in which $g'_k \ne g_i$, and $u_3(g'_k) \le u_3(M \setminus K)$.

{\bf Output:} A complete EF1 allocation.
\begin{algorithmic}[1]
\State Call {\sc Algo8} to produce an EF1 partial allocation $\mcA = (A_1, A_2, A_3)$, and $k_1, k_2, k_3$;
\State set $U = M \setminus (A_1 \cup A_2 \cup A_3)$;
\If{($k_1 > k_2$)}
	\State convert $\mcA$ into a well-defined allocation with permutation $(i_1, i_2, i_3)$;
	\State call {\sc Approx3} on $\mcA$ and $(i_1, i_2, i_3)$ to assign the items in $U$;
	\State output the final allocation $\mcA$;
\Else
	\State call {\sc Algo5} on $\mcA$ to continue to assign the items in $U = \{g'_{k_1}, \ldots, g'_{k_2}\}$;
	\State call {\sc Algo6} on $\mcA$ to output the final allocation $\mcA$.
\EndIf
\end{algorithmic}
\end{algorithm}

\begin{lemma}
\label{lemma21}
{\sc Approx9} produces a complete EF1 allocation.
\end{lemma}
\begin{proof}
When {\sc Algo8} terminates with $k_3 = k$, we have seen that the returned allocation is partial, EF1 and good,
and thus the procedure~{\sc Algo5} is applicable to continue to assign the items.
Afterwards, {\sc Algo6} is called to return the final allocation which is complete and EF1.

We consider next the case where {\sc Algo8} terminates with $k_1 > k_2$ and let $\mcA = (A_1, A_2, A_3)$ denote the allocation returned by {\sc Algo8}.
From the procedure description of {\sc Algo8} we know that in the last iteration of the while-loop, $k_1 = k_2$ and the item $g'_{k_1}$ is either assigned to $a_1$ or to $a_2$;
note from Lemma~\ref{lemma20} that before this last item $g'_{k_1}$ is assigned, the allocation was good.
We show that the bundles in $\mcA$ can be shuffled to obtain a well-defined allocation with permutation $(i_1, i_2, i_3) \in \{(1, 2, 3), (3, 1, 2)\}$.

From the procedure description of {\sc Algo8}, during the iteration of the while-loop the last item $g \in A_i$ is allocated to $a_i$, for any $i = 2, 3$,
$a_1$ was envied but $a_i$ was not envied, and thus by Lemma~\ref{lemma09} that $a_i$ envied $a_1$, i.e., $u_3(A_i \setminus g) < u_3(A_1)$.
If no item is allocated to $a_i$ during the while-loop of {\sc Algo8},
then we have $u_3(A_i \setminus g) = u_3(\emptyset) = 0$ where $g$ is either $g'_k$ for $a_2$ or void for $a_3$, and thus $u_3(A_i \setminus g) \le u_3(A_1)$.
Therefore, either way we have $u_3(A_i \setminus g) \le u_3(A_1)$.
This inequality essentially tells that if $A_1$ and $A_i$ are swapped, then $a_i$ does not strongly envy any agent.
Also, since $a_2$ and $a_3$ use the same utility function, their bundles can be swapped without violating EF1 property.
In other words, any shuffling of the three bundles gives an EF1 allocation.
We assume $u_1(A_1) \ge u_1(A_2)$, i.e., $a_1$ does not envy $a_2$, or otherwise we swap $A_1$ and $A_2$ to achieve this property.
It follows that $u_1(A_1) \ge \frac 12 u_1(K)$.
We distinguish three cases.

In the first case, $a_1$ envies $a_3$ and $a_3$ envies $a_1$, i.e., $u_1(A_1) < u_1(A_3)$ and $u_3(A_3) < u_3(A_1)$.
Consider the allocation $\mcA' = (A_3, A_2, A_1)$, in which $a_1$ does not envy $a_2$ or $a_3$, $u_1(A_3) > u_1(A_1) \ge \frac 12 u_1(K)$ and $u_3(A_3) < u_3(A_1) \le u_3(K)$.
Next, if $u_3(A_2) \ge u_3(A_1)$, then $\mcA'$ is well-defined with permutation $(1, 2, 3)$ per Definition~\ref{def12};
or if $u_3(A_2) < u_3(A_1)$, then $(A_3, A_1, A_2)$ is well-defined with permutation $(1, 2, 3)$.

In the second case, $a_1$ does not envy $a_3$.
Consider the EF1 allocation $\mcA = (A_1, A_2, A_3)$, in which $a_1$ does not envy $a_2$ or $a_3$, $u_1(A_1) \ge \frac 12 u_1(K)$ and $u_3(A_1) \le u_3(K)$.
Next, if $u_3(A_2) \ge u_3(A_3)$, then $\mcA$ is well-defined with permutation $(1, 2, 3)$;
or if $u_3(A_2) < u_3(A_3)$, then $(A_1, A_3, A_2)$ is well-defined with permutation $(1, 2, 3)$.

In the third case, $a_1$ envies $a_3$ but $a_3$ does not envy $a_1$, i.e., $u_1(A_1) < u_1(A_3)$ and $u_3(A_3) \ge u_3(A_1)$.
If $a_2$ does not envy $a_3$, then $a_2$ does not envy $a_1$ either, i.e., $u_3(A_2) \ge u_3(A_3)$ and $u_3(A_2) \ge u_3(A_1)$,
then the EF1 allocation $\mcA' = (A_3, A_2, A_1)$ satisfies $u_1(A_3) > u_1(A_1) \ge \frac 12 u_1(K)$ and $u_3(A_3) \le u_3(A_2) \le u_3(K)$,
and thus is well-defined with permutation $(1, 2, 3)$;
if $a_2$ envies $a_3$, i.e., $u_3(A_2) < u_3(A_3)$,
then the EF1 allocation $\mcA' = (A_1, A_2, A_3)$ satisfies $u_1(A_1) \ge \frac 12 u_1(K)$ and $u_3(A_1) \le u_3(K)$,
and thus is well-defined with permutation $(3, 1, 2)$.

In conclusion, the allocation $\mcA = (A_1, A_2, A_3)$ can be shuffled to become well-defined with permutation either $(1, 2, 3)$ or $(3, 1, 2)$.
It follows from Lemma~\ref{lemma13} that the returned allocation is complete and EF1.
\end{proof}

\begin{theorem}
\label{thm22}
Given the critical set $K$ for the item $g^*$ produced by the procedure {\sc Algo2} satisfying $\max_{g \in K} u_1(g) < \frac 12 u_1(K)$, $k = |K|$,
the new item order $S'$ in which $g'_k \ne g^*$, and $u_3(g'_k) \le u_3(M \setminus K)$,
{\sc Approx9} constructs a complete EF1 allocation with its total utility at least $\frac 12 OPT$.
\end{theorem}
\begin{proof}
Let $\mcA = (A_1, A_2, A_3)$ be the partial EF1 allocation produced by the procedure~\ref{Algo8} and $U$ be the unassigned item set.
We see from Lemma~\ref{lemma21} that the final allocation returned by the algorithm {\sc Approx9} is complete and EF1.
We distinguish two cases for estimating the total utility.

When {\sc Algo8} terminates with $k_3 = k$, $A_3 = M \setminus K$ and
the procedure {\sc Algo5} continues to assign the items of $U$ to produce a complete EF1 allocation denoted as $\mcA' = (A'_1, A'_2, A_3)$
where $A_1 \subseteq A'_1$ and $A_2 \subseteq A'_2$.
The procedure {\sc Algo6} outputs $\mcA'$ if $u_1(A'_1) \ge u_1(A'_2)$, or outputs the one with a larger total utility between $\mcA'$ and $(A'_2, A'_1, A_3)$ otherwise.
It follows that the total utility of the final allocation is at least $\frac 12 u_1(K) + u_3(M \setminus K)$.

When {\sc Algo8} terminates with $k_1 > k_2$, then $A_1 \cup A_2 = K$, $A_3 \subset M \setminus K$,
and we obtain a well-defined allocation $\mcA' = (A'_1, A'_2, A'_3)$ with permutation $(i_1, i_2, i_3) \in \{(1, 2, 3), (3, 1, 2)\}$ converted from $\mcA$.
From here, Lemma~\ref{lemma13} guarantees that the total utility of the final allocation is at least $\frac 12 u_1(K) + \frac 12 u_3(M \setminus K)$.

In either case, we have $SOL \ge \frac 12 u_1(K) + \frac 12 u_3(M \setminus K) \ge \frac 12 OPT$, where the last inequality is by Lemma~\ref{lemma08}.
We thus prove the theorem.
\end{proof}

\subsection{An instance to show the tightness of ratio $2$}
We provide an instance below to show that the approximation ratio $2$ of the combination of the three algorithms, {\sc Approx3}, {\sc Approx7} and {\sc Approx9},
for the case of three agents with two unnormalized utility functions is tight.
In this instance there are only five items in the order $\rho(g_1) > \rho(g_2) > \rho(g_3) > 1 > \rho(g_4) = \rho(g_5)$, with their values to the three agents listed as follows,
where $\epsilon > 0$ is a small value:
\begin{center}
\begin{tabular}{r|ccccc}
					&$g_1$			&$g_2$			&$g_3$			&$g_4$			&$g_5$\\
\hline
$a_1$				&$\epsilon$	&$1$				&$1$				&$0$				&$0$\\
$a_2, a_3$		&$0$				&$\epsilon$	&$2\epsilon$	&$\epsilon$	&$\epsilon$
\end{tabular}
\end{center}
We continue to use the same notations introduced in Section 3.
One sees that $X = \{g_1, g_2, g_3\}$, and thus $\mcA^* = (\{g_1, g_2, g_3\}, \{g_4\}, \{g_5\})$ is an optimal EF1 allocation of total utility $2 + 3\epsilon$.

For each of $g_1$ and $g_2$, $u_3(X \setminus g_i) > \frac 13 u_3(M \setminus g_i) \ge \frac 43 \epsilon$;
the critical set produced by {\sc Algo2} is $K = \{g_1, g_2, g_3\}$.
One can verify that $\max_{g \in K} u_1(g) < \frac 12 u_1(K)$, and the new item order is the same as the original preference order.
Since $u_3(g_3) = u_3(M \setminus K)$, it falls into Case 3.
In this case, {\sc Algo8} starts with the allocation $(\{g_1\}, \{g_3\}, \emptyset)$:
In the first iteration, $a_1$ is not envied by $a_2$ or $a_3$, the allocation is updated to $(\{g_1, g_2\}, \{g_3\}, \emptyset)$ and the procedure terminates with $k_1 > k_2$.
Since $(\{g_1, g_2\}, \{g_3\}, \emptyset)$ is well-defined with permutation $(1, 2, 3)$,
{\sc Approx3} calls the RR algorithm that assigns $g_4$ and $g_5$ one to each of $a_3$ and $a_2$,
achieving the final allocation either $(\{g_1, g_2\}, \{g_3, g_5\}, \{g_4\})$ or $(\{g_1, g_2\}, \{g_3, g_4\}, \{g_5\})$ of total utility $1 + 5\epsilon$.

For $g_3$, $u_3(X \setminus g_3) = \frac 13 u_3(M \setminus g_3) = \epsilon$;
the critical set produced by {\sc Algo2} is also $K = \{g_1, g_2, g_3\}$, and the new item order is $S' = \langle g_1, g_3, g_2, g_4, g_5\rangle$.
Since $u_3(g_2) < u_3(M \setminus K)$, it falls into Case 3 too.
In this case, {\sc Algo8} starts with the allocation $(\{g_1\}, \{g_2\}, \emptyset)$:
In the first iteration, $a_1$ is not envied by $a_2$ or $a_3$, the allocation is updated to $(\{g_1, g_3\}, \{g_2\}, \emptyset)$ and the procedure terminates with $k_1 > k_2$.
Since $(\{g_1, g_3\}, \{g_2\}, \emptyset)$ is well-defined with permutation $(1, 2, 3)$,
{\sc Approx3} calls the RR algorithm that assigns $g_4$ and $g_5$ one to each of $a_3$ and $a_2$,
achieving the final allocation either $(\{g_1, g_3\}, \{g_2, g_5\}, \{g_4\})$ or $(\{g_1, g_3\}, \{g_2, g_4\}, \{g_5\})$ of total utility $1 + 4\epsilon$.

For $g_4$ ($g_5$ can be identically discussed), $u_3(X \setminus g_4) = 3\epsilon > \frac 13 u_3(M \setminus g_4) = \frac 43 \epsilon$;
the critical set produced by {\sc Algo2} is $K = \{g_1, g_2, g_3, g_4\}$, and the new item order is $S' = \langle g_1, g_2, g_4, g_3, g_5\rangle$.
Since $u_3(g_3) > u_3(M \setminus K)$, it falls into Case 2.
In this case, {\sc Algo4} starts with the allocation $(\{g_1\}, \emptyset, \{g_3\})$:
In the first iteration, $a_1$ is not envied by $a_2$, the allocation is updated to $(\{g_1, g_2\}, \emptyset, \{g_3\})$;
in the second iteration, $a_1$ is envied by $a_2$, the allocation is updated to $(\{g_1, g_2\}, \{g_5\}, \{g_3\})$, and the procedure terminates with $k_1 < k = k_2$.
Since in $(\{g_1, g_2\}, \{g_5\}, \{g_3\})$ $a_1$ is not envied, {\sc Algo5} assigned $g_4$ to $a_1$, achieving the complete allocation $(\{g_1, g_2, g_4\}, \{g_5\}, \{g_3\})$.
Lastly, {\sc Algo6} confirms the final allocation is $(\{g_1, g_2, g_4\}, \{g_5\}, \{g_3\})$, of total utility $1 + 4\epsilon$.

Therefore, we have $SOL \le \max\{1 + 4\epsilon, 1 + 5\epsilon\} = 1 + 5\epsilon$.
It follows that $OPT / SOL \ge (2 + 3\epsilon) / (1 + 5\epsilon) = 2 - 7\epsilon / (1 + 5\epsilon) \to 2$, when $\epsilon$ tends to $0$.

\section{A $\frac 53$-approximation for three agents with normalized functions}
Recall that the algorithm {\sc Approx1} is a $2$-approximation for $n$ agents with normalized functions.
In this section, we consider the special case where $n = 3$, and again assume w.l.o.g. that agent $a_1$ uses the utility function $u_1(\cdot)$ and $a_2$ and $a_3$ use $u_3(\cdot)$.
We continue to assume the items are given in the preference order $\rho(g_1) \ge \rho(g_2) \ge \ldots \ge \rho(g_m)$.

We also continue to use some notations and definitions introduced earlier,
such as the sets $X = \{g \in M: u_1(g) \ge u_3(g)\}$ and $Y = M \setminus X$ defined in Eq.~(\ref{eq01}),
a good allocation defined in Definition~\ref{def02}, and so on.
Additionally, for any subset $A \subseteq M$, we define the quantity
\begin{equation}
\label{eq05}
\Delta(A) = u_1(A) - u_3(A).
\end{equation}
Since now $u_1(M) = u_3(M) = 1$, we have $\Delta(X) + \Delta(Y) = 0$ and
\begin{equation}
\label{eq06}
OPT \le u_1(X) + u_3(Y) = u_1(X) - u_3(X) + u_3(M) = \Delta(X) + 1 = 1 - \Delta(Y).
\end{equation} 

In the rest of the section we distinguish two cases on $\max_{g \in X} u_1(g)$, the most value of a single item in $X$ to agent $a_1$,
and we design two different, but similar, algorithms, respectively.

\subsection{Case 1: $\max_{g \in X} u_1(g) \le \frac 13$.}
The design idea of our algorithm~{\sc Approx10} is similar to {\sc Approx1}, with a change that,
after all the items of $X$ have been allocated, i.e., both {\em to-be-allocated} items $g_{k_1}$ and $g_{k_2}$ are in $Y$,
agent $a_t$ has the priority to receive $g_{k_2}$ if it is not envied by $a_1$.
The detailed description of the algorithm {\sc Approx10} is presented in Algorithm~\ref{Approx10}.

\begin{algorithm}
\caption{{\sc Approx10} for three agents with normalized utility functions}
\label{Approx10}
{\bf Input:} Three agents of two types and a set of $m$ indivisible items $\rho(g_1) \ge \rho(g_2) \ge \ldots \ge \rho(g_m)$,
	where $\max_{g \in X} u_1(g) \le \frac 13$.

{\bf Output:} A complete EF1 allocation.
\begin{algorithmic}[1]
\State Initialize $k_1=1$, $k_2=m$, and $A_i = \emptyset$ for every agent $a_i$;

\While {($k_1 \le k_2$)}
\State find $t = \arg \min_{i = 2, 3} u_3(A_i)$;

\If {($k_1 \le |X|$)}
	\Comment $a_1$ has the priority.
	\If {($a_1$ is not envied by $a_t$)}
	\State $A_1 = A_1 \cup \{g_{k_1}\}$ and $k_1 = k_1 + 1$;
	\Else
	\State $A_t = A_t \cup \{g_{k_2}\}$ and $k_2 = k_2 - 1$;
	\EndIf
\Else
	\Comment Change of priority: Now $a_t$ has the priority.
	\If {($a_t$ is not envied by $a_1$)}
	\State $A_t = A_t \cup \{g_{k_2}\}$ and $k_2 = k_2 - 1$;
	\Else
	\State $A_1 = A_1 \cup \{g_{k_1}\}$ and $k_1 = k_1 + 1$;
	\EndIf
\EndIf
\EndWhile

\State return the final allocation.
\end{algorithmic}
\end{algorithm}

\begin{lemma}
\label{lemma23}
{\sc Approx10} produces a good, complete and EF1 allocation.
\end{lemma}
\begin{proof}
The returned allocation by {\sc Approx10} is clearly complete from the while-loop termination condition.
We prove that at the end of each iteration of the while-loop in the algorithm, the updated allocation is good and EF1, similar to Lemma~\ref{lemma04}.
Note that the initial empty allocation is trivially good and EF1.

Assume that at the beginning of the iteration the allocation denoted as $\mcA = (A_1, A_2, A_3)$ is good and EF1;
note that the to-be-allocated items are $g_{k_1}$ and $g_{k_2}$ with $k_1 \le k_2$.

Consider the case where $k_1 > |X|$ and $a_t$ is not envied by $a_1$ (we intentionally pick this case to prove, the other three cases are symmetric),
in which the algorithm updates $A_1$ to $A_t \cup \{ g_{k_2} \}$.
By the description of {\sc Approx10}, $A_1 = \{ g_1, \ldots, g_{k_1 - 1} \}$ and $A_i \subseteq \{ g_{k_2+1}, \ldots, g_m \}$ for $i = 2, 3$.
By Definitions~\ref{def01} and~\ref{def02} and $k_1 \le k_2$, $A_1 \prec A_i \cup \{ g_{k_2} \}$ for $i = 2, 3$, and thus the updated allocation is good.

Since at the beginning of the iteration $a_t$ is not envied by $a_1$, $u_1(A_1) \ge u_3(A_t)$.
Also, by the definition of $t$, $u_3(A_i) \ge u_3(A_t)$ for the third agent $a_i$.
That is, no agent envies $a_t$ in the allocation $\mcA$, and thus does not strongly envy $a_t$ in the updated allocation (by removing the item $g_{k_2}$, if necessary).
Note that in this iteration only $a_t$ gets the item $g_{k_2}$.
Therefore, the updated allocation is EF1.
\end{proof}

Let us examine one scenario of the final allocation returned by {\sc Approx10} in the next lemma, and leave the other to the main Theorem~\ref{thm25}.

\begin{lemma}
\label{lemma24}
In the final allocation $\mcA = (A_1, A_2, A_3)$ returned by {\sc Approx10}, if $X \subseteq A_1$, then $SOL \ge \frac 23 OPT$.
\end{lemma}
\begin{proof}
If $A_1 = X$, then $\mcA$ is optimal, i.e., $SOL = OPT$.

Below we consider the scenario where $a_1$ receives some items from $Y$.
This means when {\sc Approx10} terminates, $k_1 > |X| + 1$ and $A_1 = \{g_1, \ldots, g_{k_1 - 1}\}$.
Let $Y_1 = \{g_{|X|+1}, \ldots, g_{k_1-1}\}$, that is, $A_1 = X \cup Y_1$.

We claim $\min \{\Delta(X), -\Delta(Y_1)\} \le \frac 12$.
Assume to the contrary, then $u_1(X) \ge \Delta(X) > \frac 12$ and $u_3(Y_1) \ge -\Delta(Y_1) > \frac 12$.
Consider the iteration where $a_1$ is allocated with the last item $g_{k_1-1}$.
Since at the beginning of the iteration the allocation is good, by Lemma~\ref{lemma03}, $a_1$ and $a_2$ do not envy each other, so do $a_1$ and $a_3$.
From $g_{k_1-1} \in Y$ and Lines 9--13 in the algorithm description, $a_t$ is envied by $a_1$, and thus $u_1(A_1 \setminus g_{k_1-1}) < u_1(A_2 \cup A_3)$.
It follows from $A_2 \cup A_3 = Y \setminus Y_1$ and Eq.~(\ref{eq01}) that
\[
u_1(X) < u_1(A_2 \cup A_3) < u_3(A_2 \cup A_3) = u_3(Y) - u_3(Y_1) < 1 - \frac 12 = \frac 12, \mbox{ a contradiction.}
\]

Since $A_2 \cup A_3 \subset Y$, we have $u_1(A_2 \cup A_3) < u_3(A_2 \cup A_3)$,
and therefore $SOL = u_1(A_1) + u_3(A_2 \cup A_3) > u_1(A_1) + u_1(A_2 \cup A_3) = 1$.

Using the claimed $\min \{\Delta(X), -\Delta(Y_1)\} \le \frac 12$, if $\Delta(X) \le \frac 12$, then by Eq.~(\ref{eq06}), $OPT \le \frac 32 \le \frac 32 \cdot SOL$;
if $\Delta(X) > \frac 12$, then $-\Delta(Y_1) \le \frac 12 < \Delta(X)$, and thus
$SOL = u_1(A_1) - u_3(A_1) + u_3(M) = \Delta(X) + \Delta(Y_1) + 1 > \frac 23 (\Delta(X) + 1) \ge \frac 23 OPT$, where the last inequality is by Eq.~(\ref{eq06}).
This proves the lemma.
\end{proof}

\begin{theorem}
\label{thm25}
If $\max_{g \in X} u_1(g) \le \frac 13$, then {\sc Approx10} produces a complete EF1 allocation with its total utility $SOL \ge \frac 35 OPT$.
\end{theorem}
\begin{proof}
Lemma~\ref{lemma23} states that the final allocation $\mcA = (A_1, A_2, A_3)$ returned by the algorithm is complete and EF1.
If $X \subseteq A_1$, then by Lemma~\ref{lemma24} the total utility of $\mcA$ is at least $\frac 23 OPT$.

We next consider the other scenario where $A_1 \subset X$, i.e., the algorithm {\sc Approx10} terminates with $k_1 \le |X|$.
Let $X_1 = \{g_{k_1}, \ldots, g_{|X|}\}$, then $A_1 = \{g_1, \ldots, g_{k_1-1}\} = X \setminus X_1$ and $A_2 \cup A_3 = X_1 \cup Y$.

We claim that $\Delta(X_1) \le \frac 23$, and prove it by contradiction to assume $\Delta(X_1) > \frac 23$.
It follows that $u_1(X_1) > \frac 23$ and thus $u_1(A_1) \le u_1(M \setminus X_1) < \frac 13$.

From Lines 4--8 in the description of the algorithm, the item $g_{k_1}$ is assigned to agent $a_t$ in the last iteration of the while-loop
(in which $k_2 = k_1$ at the beginning of the iteration and $k_2$ is decremented afterwards leading to the termination condition $k_2 < k_1$).
That is, $g_{k_1}$ is the last item received by $a_t$.
Denote the other agent as $a_i$, i.e., $\{t, i\} = \{2, 3\}$.
Since $a_1$ is envied at the beginning of the iteration, $a_1$ is envied by $a_t$;
further because the allocation is good, $a_1$ does not envy $a_t$.
To summarize,
\begin{equation}
\label{eq07}
u_3(A_t \setminus g_{k_1}) \le u_3(A_i), \mbox{ }
u_3(A_t \setminus g_{k_1}) < u_3(A_1), \mbox{ and }
u_1(A_1) \ge u_1(A_t \setminus g_{k_1}).
\end{equation}
It follows from $A_1 \subseteq X$ that 
\begin{equation}
\label{eq08}
u_3(A_t \setminus g_{k_1}) < u_3(A_1) \le u_1(A_1) < \frac 13.
\end{equation}
One sees that if agent $a_i$ also receives some items from $X_1$, then the above argument applies too for the last item $a_i$ receives, denoted as $g_\ell$,
such that Eq.~(\ref{eq08}) holds, i.e., $u_3(A_i \setminus g_\ell) < \frac 13$.

Subsequently, if $A_2 \cap X_1 \ne \emptyset$ and $A_3 \cap X_1 \ne \emptyset$, then
\[
\Delta(X_1) \le \Delta(X) = - \Delta(Y) \le u_3(A_t \setminus g_{k_1}) + u_3(A_i \setminus g_\ell) < \frac 23, \mbox{ a contradiction.}
\]
If $X_1 \subseteq A_t$, then by Eq.~(\ref{eq07}),
\[
u_1(g_{k_1}) \ge u_1(A_t) - u_1(A_1) \ge u_1(X_1) - u_1(A_1) > \frac 13, \mbox{ contradicting to $\max_{g \in X} u_1(g) \le \frac 13$.}
\]
We have thus proved the claim that $\Delta(X_1) \le \frac 23$.

We linearly combine $\Delta(X_1) \le \frac 23$ and $\Delta(X_1) \le \Delta(X)$ to have
$\Delta(X_1) \le \frac 35 \cdot \frac 23 + \frac 25 \cdot \Delta(X) = \frac 25 + \frac 25 \Delta(X)$, and thus by Eq.~(\ref{eq06})
\[
SOL = u_1(X \setminus X_1) + u_3(X_1 \cup Y) = 1 + \Delta(X) - \Delta(X_1) \ge \frac 35 (1 + \Delta(X)) \ge \frac 35 OPT.
\]
This proves the theorem.
\end{proof}

\subsection{Case 2: $\max_{g \in X} u_1(g) > \frac 13$.}
In this case, we let $g^* = \arg \max_{g \in X} u_1(g)$;
thus $u_1(g^*) > \frac 13$, which is so big that it is assigned to agent $a_1$ immediately.
On the other hand, we can use $g^*$ to bound $OPT$ as follows:
\begin{equation}
\label{eq09}
OPT \le u_1(X) + u_3(Y) \le u_1(X) + u_3(M) - u_3(g^*) \le 2 - u_3(g^*) < \frac 53 + \Delta(g^*).
\end{equation}
Our algorithm for this case is very similar to {\sc Approx10}, with two changes:
One is the starting allocation set to be $(\{g^*\}, \emptyset, \emptyset)$,
and the other is, when $k_2 = |X|$, the while-loop terminates and the algorithm switches to the Envy-Cycle Elimination (ECE) algorithm to assign the rest of items.
The detailed description of the algorithm, denoted as {\sc Approx11}, is presented in Algorithm~\ref{Approx11}.

\begin{algorithm}
\caption{{\sc Approx11} for three agents with normalized utility functions}
\label{Approx11}
{\bf Input:} Three agents of two types and a set of $m$ indivisible items $\rho(g_1) \ge \rho(g_2) \ge \ldots \ge \rho(g_m)$,
	where $g^* = \arg\max_{g \in X} u_1(g)$ and $u_1(g^*) > \frac 13$.

{\bf Output:} A complete EF1 allocation.
\begin{algorithmic}[1]
\State Initialize $k_1=1$, $k_2=m$, and $\mcA = (\{g^*\}, \emptyset, \emptyset)$;
\State if $g_{k_1} = g^*$, then $k_1 = k_1 + 1$;

\While {($k_1 \le k_2$ and $k_2 > |X|$)}
\State find $t = \arg \min_{i = 2, 3} u_3(A_i)$;

\If {($k_1 \le |X|$)}
	\If {($a_1$ is not envied by $a_t$)}
	\State $A_1 = A_1 \cup \{g_{k_1}\}$, $k_1 = k_1 + 1$, if $g_{k_1} = g^*$ then $k_1 = k_1 + 1$; 
	\Else
	\State $A_t = A_t \cup \{g_{k_2}\}$ and $k_2 = k_2 - 1$;
	\EndIf
\Else
	\If {($a_t$ is not envied by $a_1$)}
	\State $A_t = A_t \cup \{g_{k_2}\}$ and $k_2 = k_2 - 1$;
	\Else
	\State $A_1 = A_1 \cup \{g_{k_1}\}$, $k_1 = k_1 + 1$, if $g_{k_1} = g^*$ then $k_1 = k_1 + 1$; 
	\EndIf
\EndIf
\EndWhile

\If {($k_2 = |X|$)}
	\State call the ECE algorithm on $\mcA$ to continue to assign the items in $\{g_{k_1}, \ldots, g_{k_2}\} \setminus \{g^*\}$;
\EndIf
\State return the final allocation.
\end{algorithmic}
\end{algorithm}

\begin{theorem}
\label{thm26}
If $\max_{g \in X} u_1(g) > \frac 13$, then {\sc Approx11} produces a complete EF1 allocation with its total utility $SOL \ge \frac 35 OPT$,
and the approximation ratio $\frac 35$ is tight.
\end{theorem}
\begin{proof}
We distinguish the two termination condition of the while-loop in {\sc Approx11}.

In the first scenario, the while-loop terminates at $k_1 > k_2$, that is, all the items are assigned and
$k_2 \ge |X|$ implying $X \subseteq A_1$ in the returned allocation $\mcA = (A_1, A_2, A_3)$.
Note that the while-loop body is the same as the while-loop body in {\sc Approx10}.
It follows from Lemma~\ref{lemma23} that $\mcA$ is good, complete and EF1, and from Lemma~\ref{lemma24} that its total utility is $SOL \ge \frac 23 OPT$.

In the other scenario, the while-loop terminates at $k_1 \le k_2 = |X|$, that is, all the items of $Y$ have been assigned to $a_2$ and $a_3$ (i.e., $Y \subseteq A_2 \cup A_3$),
and the unassigned items are $g_{k_1}, \ldots, g_{k_2}$ excluding $g^*$.
By Lemma~\ref{lemma23}, the achieved allocation $\mcA$ is good and EF1.
Since $g^* \in A_1$, we have $\Delta(A_1) \ge \Delta(g^*)$.

We claim that the ECE algorithm does not decrease $\Delta(A_1)$.
To prove the claim, we assume in an iteration of the ECE algorithm, with the allocation $\mcA = (A_1, A_2, A_3)$ at the beginning,
agent $a_1$ gets the bundle $A_2$ or $A_3$ due to the existence of an envy cycle.
Further assume w.l.o.g. that $a_1$ gets bundle $A_2$, which means $a_1$ envies $a_2$ and the envy cycle is either $(1 \to 2 \to 3 \to 1)$ or $(1 \to 2 \to 1)$.
Either way, $u_1(A_1) < u_1(A_2)$ and $u_3(A_2) (< u_3(A_3)) < u_3(A_1)$.
Therefore, $\Delta(A_1) < \Delta(A_2)$.
This proves the claim.

The final allocation $\mcA = (A_1, A_2, A_3)$ returned from the ECE algorithm has its total utility
\[
SOL = u_1(A_1) + u_3(A_2 \cup A_3) \ge \Delta(A_1) + u_3(M) \ge \Delta(g^*) + 1 > \frac 35 OPT,
\]
where the last inequality is by Eq.~(\ref{eq09}).
We thus prove that the approximation ratio of {\sc Approx11} is at least $\frac 35$.

We next provide an instance below to show that this approximation ratio is tight.
In this instance, there are five items given in the order $\rho(g_1) > \rho(g_2) > \rho(g_3) > 1 > \rho(g_4) = \rho(g_5)$,
with their values to the three agents listed as follows, where $\epsilon> 0$ is a small value:
\begin{center}
\begin{tabular}{r|ccccc}
					&$g_1$			&$g_2$							&$g_3$							&$g_4$							&$g_5$\\
\hline
$a_1$				&$\frac 13$	&$\frac 13 - \epsilon$	&$\frac 13 + \epsilon$	&$0$								&$0$\\
$a_2, a_3$		&$\epsilon$	&$\epsilon$					&$\frac 13$					&$\frac 13 - \epsilon$	&$\frac 13 - \epsilon$
\end{tabular}
\end{center}
One sees that the instance is normalized, $X = \{g_1, g_2, g_3\}$, and $A^* = (\{g_1, g_2\}, \{g_3, g_4\}, \{g_5\})$ is an optimal EF1 allocation of total utility $\frac 53 - 3 \epsilon$.

Since $u_1(g_3) > \frac 13$, {\sc Approx11} is executed which starts with the allocation $(\{g_3\}, \emptyset , \emptyset)$:
In the first iteration, $a_1$ 	is envied by $a_2$ and $a_3$, and the allocation is updated to $(\{g_3\}, \emptyset, \{g_5\})$ or $(\{g_3\}, \{g_5\}, \emptyset)$;
in the second iteration, $a_1$ and $a_3$ (resp., $a_1$ and $a_2$) are envied by $a_2$ (resp., $a_3$), and the allocation is updated to $\mcA = (\{g_3\}, \{g_4\}, \{g_5\})$;
the procedure terminates with $k_2 = |X| = 3$.

Next, the ECE algorithm is called on $\mcA$ to continue to assign the items $g_1$ and $g_2$:
In the first iteration, since in $(\{g_3\}, \{g_4\}, \{g_5\})$ only $a_1$ is still envied by $a_2$ and $a_3$, $a_2$ and $a_3$ are not envied,
the ECE algorithm assigns $g_1$ to $a_2$ (or to $a_3$, symmetrically) resulting in the updated allocation $(\{g_3\}, \{g_1, g_4\}, \{g_5\})$;
in the second iteration, the ECE algorithm assigns the last item $g_2$ to $a_3$, ending at the complete allocation $\mcA = (\{g_3\}, \{g_1, g_4\}, \{g_2, g_5\})$.
This final allocation $\mcA = (\{g_3\}, \{g_1, g_4\}, \{g_2, g_5\})$ has total utility $1$.
Therefore, we have $OPT / SOL = (\frac 53 - 3 \epsilon) / 1 \to \frac 53$, when $\epsilon$ tends to $0$.
\end{proof}

\section{Conclusion}
Fair division of indivisible goods is an interesting problem which has received a lot of studies,
from multiple research communities including artificial intelligence and theoretical computer science.
Finding an EF1 allocation to maximize the utilitarian social welfare from the perspective of approximation algorithms emerges most recently.
Despite EF1 being one of the most simplest fairness criteria, the design and analysis of approximation algorithms for this problem in the general case seems challenging~\cite{BBS20,BLL25}.
In this paper, we focused on the special case where agents are of only two types,
and we hope our work may shed lights on and inspire more studies for the general case.
For this special case, we presented a $2$-approximation algorithm for any number of agents with normalized utility functions;
by the lower bound of $\frac {4n}{3n+1}$ from \cite{BLL25}, our result shows the problem in this special case is APX-complete.
When there are only three agents, we presented an improved $\frac 53$-approximation algorithm.
When there are only three agents but the utility functions are unnormalized,
we presented a tight $2$-approximation algorithm which is the best possible by the lower bound of $\frac {1 + \sqrt{4n-3}}2$ from \cite{BLL25}.
In all three algorithms, we demonstrate the use of the item {\em preference} (Definition~\ref{def01}) order, which can be explored further for improved algorithms.

We remark that the lower bound of $\frac {4n}{3n+1}$ on the approximability in~\cite{BLL25} is proven for a more restricted case
where the two utility functions are normalized, agent $a_1$ uses a utility function and all the other agents use the other function.
It would be interesting to narrow the gap for either the special case we study 
or this more restricted case;
we expect some improved lower bounds or some approximation ratios as functions in $n_1$ and $n_2$,
where $n_1$ agents use a common utility function and the other $n_2 = n - n_1$ use the second utility function.



\end{document}